\documentclass[12pt,reqno,a4paper]{amsart}
\usepackage{amsmath,amssymb,amsfonts,amscd}
\usepackage[mathscr]{eucal}
\usepackage{comment}
\usepackage{color}
\usepackage{hyperref,graphicx}

\usepackage{tikz}
\usepackage{tikz-cd}
\usetikzlibrary{positioning}
\usepackage{arydshln}
\usepackage{yfonts}
\newcommand{\Gr}{\mathop{Gr}}

\newcommand{\ZZ}{{\mathbb Z}_2}

\newcommand{\aund}{{\underline{a}}}

\newcommand{\bund}{{\underline{b}}}

\newcommand{\hmund}{{\underline{\hat \mu}}}
\newcommand{\hnund}{{\underline{\hat \nu}}}

\DeclareMathOperator{\Ber}{Ber}

\DeclareMathOperator{\adj}{adj}

\DeclareMathOperator{\GL}{GL}
\DeclareMathOperator{\cofactor}{cofactor}

\newcommand{\dder}[3]{{\frac{\partial^2 {#1}}{\partial {#2}\partial {#3}}}}

\newcommand{\RR}{\mathbb R}

\newcommand{\al}{{\alpha}}
\newcommand{\be}{{\beta}}

\renewcommand{\t}{{\theta}}

\newcommand{\wed}{\wedge}

\renewcommand{\L}{{\Lambda}}

\newcommand{\itt}{{\tilde\imath}}
\newcommand{\jt}{{\tilde\jmath}}

\newcommand{\at}{{\tilde a}}
\newcommand{\bt}{{\tilde b}}
\newcommand{\ct}{{\tilde c}}
\newcommand{\dt}{{\tilde d}}

\DeclareMathOperator{\pl}{\text{\textswab{pl\"uck}}}

\DeclareMathOperator{\Pl}{\text{\textswab{Pl\"uck}}}

\newcommand{\up}{\boldsymbol{u}}

\newcommand{\ep}{\boldsymbol{e}}

\newcommand{\hnu}{\hat\nu}
\newcommand{\hmu}{\hat\mu}

\newcommand{\hla}{\hat\lambda}

\newcommand{\pfinth}{P_{1,-1}^{\text{\emph{fin}}}\left(\L^{r|s}(V)\oplus \L^{s|r}(\Pi V)\right)}

\newcommand{\pessth}{P^{\text{\emph{ess}}}\left(\L^{r}(V)\right)}

\newcommand{\cir}{\begin{tikzpicture} 
\filldraw[color=black, fill=black, very thick](-7, -1) circle (0.05);
\end{tikzpicture}}
\newcommand{\cirr}{\begin{tikzpicture} 
\filldraw[color=red, fill=red, very thick](-7, -1) circle (0.05);
\end{tikzpicture}}

\newcommand{\pent}{\begin{tikzpicture} 
\path (90:1cm) coordinate (P1);
\path (162:1cm) coordinate (P2);
\path (234:1cm) coordinate (P3);
\path (306:1cm) coordinate (P4);
\path (18:1cm) coordinate (P5);
\draw (P1) -- (P2) -- (P3) -- (P4) -- (P5) -- cycle;
\node [anchor=south] at (P1)   {$1$};
\node [anchor=east] at (P2)   {$2$};
\node [anchor=east] at (P3)   {$3$};
\node [anchor=west] at (P4)   {$4$};
\node [anchor=west] at (P5)   {$5$};
\path [draw, very thick] (P1) -- (P3) -- cycle;
\path [draw, very thick] (P1) -- (P4) -- cycle;
\filldraw[red] (90:1cm) circle (2pt);
\filldraw[red] (306:1cm) circle (2pt);
\end{tikzpicture}}
\newcommand{\ocot}
{\begin{tikzpicture}
\path (90:1cm) coordinate (P1);
\path (162:1cm) coordinate (P2);
\path (234:1cm) coordinate (P3);
\path (306:1cm) coordinate (P4);
\path (18:1cm) coordinate (P5);
\draw (P1) -- (P2) -- (P3) -- (P4) -- (P5) -- cycle;
\node [anchor=south] at (P1)   {$1$};
\node [anchor=east] at (P2)   {$2$};
\node [anchor=east] at (P3)   {$3$};
\node [anchor=west] at (P4)   {$4$};
\node [anchor=west] at (P5)   {$5$};
\path [draw, very thick] (P1) -- (P4) -- cycle;
\path [draw, very thick] (P1) -- (P3) -- cycle;
\filldraw[red] (90:1cm) circle (2pt);
\filldraw[red] (306:1cm) circle (2pt);
\end{tikzpicture}}

\newcommand{\otoc}
{\begin{tikzpicture}
\path (90:1cm) coordinate (P1);
\path (162:1cm) coordinate (P2);
\path (234:1cm) coordinate (P3);
\path (306:1cm) coordinate (P4);
\path (18:1cm) coordinate (P5);
\draw (P1) -- (P2) -- (P3) -- (P4) -- (P5) -- cycle;
\node [anchor=south] at (P1)   {$1$};
\node [anchor=east] at (P2)   {$2$};
\node [anchor=east] at (P3)   {$3$};
\node [anchor=west] at (P4)   {$4$};
\node [anchor=west] at (P5)   {$5$};
\path [draw, very thick] (P1) -- (P4) -- cycle;
\path [draw, very thick] (P1) -- (P3) -- cycle;
\filldraw[red] (90:1cm) circle (2pt); 
\filldraw[red] (234:1cm) circle (2pt);
\end{tikzpicture}}
\newcommand{\pentodd}{
\begin{tikzpicture}
\node (A) {\ocot};
\node [right=3cm =of A] (B) {\otoc};
\path[draw] (A) -- (B) [dashed, color=red, very thick];
\end{tikzpicture}}
\newcommand{\penteven}{
\begin{tikzpicture}
\node (A) {\ocot};
\node [right=3cm =of A] (B) {\tpot};
\path[draw] (A) -- (B) [color=black, very thick];
\end{tikzpicture}}
\newcommand{\tpot}
{\begin{tikzpicture}
\path (90:1cm) coordinate (P1);
\path (162:1cm) coordinate (P2);
\path (234:1cm) coordinate (P3);
\path (306:1cm) coordinate (P4);
\path (18:1cm) coordinate (P5);
\draw (P1) -- (P2) -- (P3) -- (P4) -- (P5) -- cycle;
\node [anchor=south] at (P1)   {$1$};
\node [anchor=east] at (P2)   {$2$};
\node [anchor=east] at (P3)   {$3$};
\node [anchor=west] at (P4)   {$4$};
\node [anchor=west] at (P5)   {$5$};
\path [draw, very thick] (P3) -- (P5) -- cycle;
\path [draw, very thick] (P1) -- (P3) -- cycle;
\filldraw[red] (234:1cm) circle (2pt); 
\filldraw[red] (18:1cm) circle (2pt); 
\end{tikzpicture}}
\newcommand{\ottp}
{\begin{tikzpicture}
\path (90:1cm) coordinate (P1);
\path (162:1cm) coordinate (P2);
\path (234:1cm) coordinate (P3);
\path (306:1cm) coordinate (P4);
\path (18:1cm) coordinate (P5);
\draw (P1) -- (P2) -- (P3) -- (P4) -- (P5) -- cycle;
\node [anchor=south] at (P1)   {$1$};
\node [anchor=east] at (P2)   {$2$};
\node [anchor=east] at (P3)   {$3$};
\node [anchor=west] at (P4)   {$4$};
\node [anchor=west] at (P5)   {$5$};
\path [draw, very thick] (P3) -- (P5) -- cycle;
\path [draw, very thick] (P1) -- (P3) -- cycle;
\filldraw[red] (234:1cm) circle (2pt); 
\filldraw[red] (90:1cm) circle (2pt); 
\end{tikzpicture}}
\newcommand{\dptp}
{\begin{tikzpicture}
\path (90:1cm) coordinate (P1);
\path (162:1cm) coordinate (P2);
\path (234:1cm) coordinate (P3);
\path (306:1cm) coordinate (P4);
\path (18:1cm) coordinate (P5);
\draw (P1) -- (P2) -- (P3) -- (P4) -- (P5) -- cycle;
\node [anchor=south] at (P1)   {$1$};
\node [anchor=east] at (P2)   {$2$};
\node [anchor=east] at (P3)   {$3$};
\node [anchor=west] at (P4)   {$4$};
\node [anchor=west] at (P5)   {$5$};
\path [draw, very thick] (P2) -- (P5) -- cycle;
\path [draw, very thick] (P3) -- (P5) -- cycle;
\filldraw[red] (162:1cm) circle (2pt); 
\filldraw[red] (18:1cm) circle (2pt); 
\end{tikzpicture}}

\newcommand{\tpdp}
{\begin{tikzpicture}
\path (90:1cm) coordinate (P1);
\path (162:1cm) coordinate (P2);
\path (234:1cm) coordinate (P3);
\path (306:1cm) coordinate (P4);
\path (18:1cm) coordinate (P5);
\draw (P1) -- (P2) -- (P3) -- (P4) -- (P5) -- cycle;
\node [anchor=south] at (P1)   {$1$};
\node [anchor=east] at (P2)   {$2$};
\node [anchor=east] at (P3)   {$3$};
\node [anchor=west] at (P4)   {$4$};
\node [anchor=west] at (P5)   {$5$};
\path [draw, very thick] (P2) -- (P5) -- cycle;
\path [draw, very thick] (P3) -- (P5) -- cycle;
\filldraw[red] (234:1cm) circle (2pt); 
\filldraw[red] (18:1cm) circle (2pt); 
\end{tikzpicture}}

\newcommand{\dcdp}
{\begin{tikzpicture}
\path (90:1cm) coordinate (P1);
\path (162:1cm) coordinate (P2);
\path (234:1cm) coordinate (P3);
\path (306:1cm) coordinate (P4);
\path (18:1cm) coordinate (P5);
\draw (P1) -- (P2) -- (P3) -- (P4) -- (P5) -- cycle;
\node [anchor=south] at (P1)   {$1$};
\node [anchor=east] at (P2)   {$2$};
\node [anchor=east] at (P3)   {$3$};
\node [anchor=west] at (P4)   {$4$};
\node [anchor=west] at (P5)   {$5$};
\path [draw, very thick] (P2) -- (P5) -- cycle;
\path [draw, very thick] (P2) -- (P4) -- cycle;
\filldraw[red] (162:1cm) circle (2pt); 
\filldraw[red] (306:1cm) circle (2pt); 
\end{tikzpicture}}

\newcommand{\dpdc}
{\begin{tikzpicture}
\path (90:1cm) coordinate (P1);
\path (162:1cm) coordinate (P2);
\path (234:1cm) coordinate (P3);
\path (306:1cm) coordinate (P4);
\path (18:1cm) coordinate (P5);
\draw (P1) -- (P2) -- (P3) -- (P4) -- (P5) -- cycle;
\node [anchor=south] at (P1)   {$1$};
\node [anchor=east] at (P2)   {$2$};
\node [anchor=east] at (P3)   {$3$};
\node [anchor=west] at (P4)   {$4$};
\node [anchor=west] at (P5)   {$5$};
\path [draw, very thick] (P2) -- (P5) -- cycle;
\path [draw, very thick] (P2) -- (P4) -- cycle;
\filldraw[red] (162:1cm) circle (2pt); 
\filldraw[red] (18:1cm) circle (2pt); 
\end{tikzpicture}}

\newcommand{\ocdc}
{\begin{tikzpicture}
\path (90:1cm) coordinate (P1);
\path (162:1cm) coordinate (P2);
\path (234:1cm) coordinate (P3);
\path (306:1cm) coordinate (P4);
\path (18:1cm) coordinate (P5);
\draw (P1) -- (P2) -- (P3) -- (P4) -- (P5) -- cycle;
\node [anchor=south] at (P1)   {$1$};
\node [anchor=east] at (P2)   {$2$};
\node [anchor=east] at (P3)   {$3$};
\node [anchor=west] at (P4)   {$4$};
\node [anchor=west] at (P5)   {$5$};
\path [draw, very thick] (P1) -- (P4) -- cycle;
\path [draw, very thick] (P2) -- (P4) -- cycle;
\filldraw[red] (90:1cm) circle (2pt); 
\filldraw[red] (306:1cm) circle (2pt); 
\end{tikzpicture}}

\newcommand{\dcoc}
{\begin{tikzpicture}
\path (90:1cm) coordinate (P1);
\path (162:1cm) coordinate (P2);
\path (234:1cm) coordinate (P3);
\path (306:1cm) coordinate (P4);
\path (18:1cm) coordinate (P5);
\draw (P1) -- (P2) -- (P3) -- (P4) -- (P5) -- cycle;
\node [anchor=south] at (P1)   {$1$};
\node [anchor=east] at (P2)   {$2$};
\node [anchor=east] at (P3)   {$3$};
\node [anchor=west] at (P4)   {$4$};
\node [anchor=west] at (P5)   {$5$};
\path [draw, very thick] (P1) -- (P4) -- cycle;
\path [draw, very thick] (P2) -- (P4) -- cycle;
\filldraw[red] (162:1cm) circle (2pt); 
\filldraw[red] (306:1cm) circle (2pt); 
\end{tikzpicture}}

\newcommand{\sqot}
{\begin{tikzpicture}
\path (45:1cm) coordinate (P1);
\path (135:1cm) coordinate (P2);
\path (225:1cm) coordinate (P3);
\path (315:1cm) coordinate (P4);
\draw (P1) -- (P2) -- (P3) -- (P4) -- cycle;
\node [anchor=west] at (P1)   {$1$};
\node [anchor=east] at (P2)   {$2$};
\node [anchor=east] at (P3)   {$3$};
\node [anchor=west] at (P4)   {$4$};
\path [draw, very thick] (P1) -- (P3) -- cycle;
\filldraw[red] (45:1cm) circle (2pt); 
\filldraw[red] (225:1cm) circle (2pt); 
\end{tikzpicture}}

\newcommand{\sqdc}
{\begin{tikzpicture}
\path (45:1cm) coordinate (P1);
\path (135:1cm) coordinate (P2);
\path (225:1cm) coordinate (P3);
\path (315:1cm) coordinate (P4);
\draw (P1) -- (P2) -- (P3) -- (P4) -- cycle;
\node [anchor=west] at (P1)   {$1$};
\node [anchor=east] at (P2)   {$2$};
\node [anchor=east] at (P3)   {$3$};
\node [anchor=west] at (P4)   {$4$};
\path [draw, very thick] (P2) -- (P4) -- cycle;
\filldraw[red] (135:1cm) circle (2pt); 
\filldraw[red] (315:1cm) circle (2pt); 
\end{tikzpicture}}

\newcommand{\pentaexch}{%
\begin{tikzpicture}[thick,scale=0.6, every node/.style={scale=0.6}]
\node (A) {\ocot};
\node (B) [below left =of A] {\tpot};
\node (C) [below  =of B] {\ottp};
\node (D) [below  =of C] {\dptp};
\node (E) [below right =of D] {\tpdp};
\node (F) [right =of E] {\dcdp};
\node (G) [above right =of F] {\dpdc};
\node (H) [above =of G] {\ocdc};
\node (I) [above =of H] {\dcoc};
\node (K) [right =of A] {\otoc};
\path[draw] (A) -- (K) [dashed, color=red, very thick];
\path[draw] (I) -- (H) [dashed, color=red, very thick];
\path[draw] (F) -- (G) [dashed, color=red, very thick];
\path[draw] (D) -- (E) [dashed, color=red, very thick];
\path[draw] (B) -- (C) [dashed, color=red, very thick];
\path[draw] (A) -- (B) [color=black, very thick];
\path[draw] (K) -- (I) [color=black, very thick];
\path[draw] (C) -- (D) [color=black, very thick];
\path[draw] (F) -- (E) [color=black, very thick];
\path[draw] (H) -- (G) [color=black, very thick];
\end{tikzpicture}}

\newcommand{\pentaexchsupersix}{%
\begin{tikzpicture}[scale=0.7]
		\node (0) at (-7, -1) {\cir};
		\node (1) at (-6.5, -2) {\cir};
		\node (2) at (-6, -1) {\cir};
		\node (3) at (-8, 0) {\cir};
		\node (4) at (-5, 0) {\cir};
		\node (5) at (-9, 0.5) {\cir};
		\node (6) at (-4, 0.5) {\cir};
		\node (7) at (-8, 1) {\cir};
		\node (8) at (-5, 1) {\cir};
		\node (9) at (-7, 2) {\cir};
		\node (10) at (-6, 2) {\cir};
		\node (11) at (-6.5, 3) {\cir};
		\node (12) at (-1, -1) {\cir};
		\node (13) at (-0.5, -2) {\cir};
		\node (14) at (0, -1) {\cir};
		\node (15) at (-2, 0) {\cir};
		\node (16) at (1, 0) {\cir};
		\node (17) at (-3, 0.5) {\cir};
		\node (18) at (2, 0.5) {\cir};
		\node (19) at (-2, 1) {\cir};
		\node (20) at (1, 1) {\cir};
		\node (21) at (-1, 2) {\cir};
		\node (22) at (0, 2) {\cir};
		\node (23) at (-0.5, 3) {\cir};
		\node (24) at (-4, 4) {\cir};
		\node (25) at (-3, 4) {\cir};
		\node (26) at (-3.5, 5) {\cir};
		\node (27) at (-4, -3) {\cir};
		\node (28) at (-3, -3) {\cir};
		\node (29) at (-3.5, -4) {\cir};
		\node (30) at (-10, 0.5) {\cir};
		\node (31) at (-3.5, 6) {\cir};
		\node (32) at (3, 0.5) {\cir};
		\node (33) at (-3.5, -5) {\cir};
		\node (34) at (-11, 1) {\cir};
		\node (35) at (-11, 0) {\cir};
		\node (36) at (-4, 7) {\cir};
		\node (37) at (-3, 7) {\cir};
		\node (38) at (4, 1) {\cir};
		\node (39) at (4, 0) {\cir};
		\node (40) at (-4, -6) {\cir};
		\node (41) at (-3, -6) {\cir};
		\draw[style=dashed, color=red] (0.center) to (1.center);
		\draw[style=dashed, color=red] (1.center) to (2.center);
		\draw (3.center) to (0.center);
		\draw (2.center) to (4.center);
		\draw[style=dashed, color=red] (5.center) to (3.center);
		\draw[style=dashed, color=red] (5.center) to (7.center);
		\draw (7.center) to (9.center);
		\draw[style=dashed, color=red] (9.center) to (11.center);
		\draw[style=dashed, color=red] (11.center) to (10.center);
		\draw (10.center) to (8.center);
		\draw[style=dashed, color=red] (8.center) to (6.center);
		\draw[style=dashed, color=red] (6.center) to (4.center);
		\draw[style=dashed, color=red] (12.center) to (13.center);
		\draw[style=dashed, color=red] (13.center) to (14.center);
		\draw (15.center) to (12.center);
		\draw (14.center) to (16.center);
		\draw[style=dashed, color=red] (17.center) to (15.center);
		\draw[style=dashed, color=red] (17.center) to (19.center);
		\draw (19.center) to (21.center);
		\draw (21.center) to (23.center);
		\draw[style=dashed, color=red] (23.center) to (22.center);
		\draw[style=dashed, color=red] (22.center) to (20.center);
		\draw[style=dashed, color=red] (20.center) to (18.center);
		\draw[style=dashed, color=red] (18.center) to (16.center);
		\draw (6.center) to (17.center);
		\draw (11.center) to (24.center);
		\draw (23.center) to (25.center);
		\draw[style=dashed, color=red] (24.center) to (26.center);
		\draw[style=dashed, color=red] (26.center) to (25.center);
		\draw (1.center) to (27.center);
		\draw[style=dashed, color=red] (27.center) to (29.center);
		\draw[style=dashed, color=red] (29.center) to (28.center);
		\draw (28.center) to (13.center);
		\draw (30.center) to (5.center);
		\draw (31.center) to (26.center);
		\draw (18.center) to (32.center);
		\draw (29.center) to (33.center);
		\draw[style=dashed, color=red] (34.center) to (30.center);
		\draw[style=dashed, color=red] (35.center) to (30.center);
		\draw[style=dashed, color=red] (33.center) to (40.center);
		\draw[style=dashed, color=red] (33.center) to (41.center);
		\draw[style=dashed, color=red] (36.center) to (31.center);
		\draw[style=dashed, color=red] (31.center) to (37.center);
		\draw[style=dashed, color=red] (32.center) to (38.center);
		\draw[style=dashed, color=red] (32.center) to (39.center);
		\draw (34.center) to (36.center);
		\draw (37.center) to (38.center);
		\draw (35.center) to (40.center);
		\draw (41.center) to (39.center);	
\end{tikzpicture}}

\newcommand{\pentaexchsix}{%
\begin{tikzpicture}
		\node (0) at (-3, 0) {\cir};
		\node (1) at (-2, 0) {\cir};
		\node (2) at (-4, 0) {\cir};
		\node (3) at (-3.5, 1) {\cir};
		\node (4) at (-3.5, -1) {\cir};
		\node (5) at (-1.5, 1) {\cir};
		\node (6) at (-1, 0) {\cir};
		\node (7) at (-1.5, -1) {\cir};
		\node (8) at (-2.5, -2) {\cir};
		\node (9) at (-2.5, 2) {\cir};
		\node (10) at (-5, 0) {\cir};
		\node (11) at (-2.5, -3.25) {\cir};
		\node (12) at (-2.5, 3) {\cir};
		\node (13) at (0, 0) {\cir};
		\draw (0.center) to (1.center);
		\draw (10.center) to (2.center);
		\draw (6.center) to (13.center);
		\draw (2.center) to (3.center);
		\draw (3.center) to (0.center);
		\draw (2.center) to (4.center);
		\draw (4.center) to (0.center);
		\draw (1.center) to (5.center);
		\draw (5.center) to (6.center);
		\draw (6.center) to (7.center);
		\draw (7.center) to (1.center);
		\draw (3.center) to (9.center);
		\draw (9.center) to (5.center);
		\draw (4.center) to (8.center);
		\draw (8.center) to (7.center);
		\draw (8.center) to (11.center);
		\draw (12.center) to (9.center);
		\draw (10.center) to (12.center);
		\draw (12.center) to (13.center);
		\draw (10.center) to (11.center);
		\draw (11.center) to (13.center);
\end{tikzpicture}}

\newcommand{\pentaexchfiveboth}{%
\begin{tikzpicture}
		\node (0) at (-3, 0) {\cir};
		\node (1) at (-2, 1) {\cir};
		\node (2) at (-1, 0) {\cir};
		\node (3) at (-2.5, -1.25) {\cir};
		\node (4) at (-1.5, -1.25) {\cir};
		\node (5) at (1, 0.5) {\cir};
		\node (6) at (1, -0.5) {\cir};
		\node (7) at (2.5, 2) {\cir};
		\node (8) at (3.5, 2) {\cir};
		\node (9) at (5, 0.5) {\cir};
		\node (10) at (5, -0.5) {\cir};
		\node (11) at (1.5, -2) {\cir};
		\node (12) at (2, -2.5) {\cir};
		\node (13) at (4, -2.5) {\cir};
		\node (14) at (4.5, -2) {\cir};
		\draw (0.center) to (1.center);
		\draw (1.center) to (2.center);
		\draw (2.center) to (4.center);
		\draw (3.center) to (4.center);
		\draw (3.center) to (0.center);
		\draw[style=dashed, color=red] (5.center) to (6.center);
		\draw[style=dashed, color=red, very thick] (9.center) to (10.center);
		\draw[style=dashed, color=red, very thick] (11.center) to (12.center);
		\draw[style=dashed, color=red, very thick] (13.center) to (14.center);
		\draw[style=dashed, color=red, very thick] (7.center) to (8.center);
		\draw (7.center) to (5.center);
		\draw (6.center) to (11.center);
		\draw (12.center) to (13.center);
		\draw (14.center) to (10.center);
		\draw (8.center) to (9.center);	
\end{tikzpicture}}

\newcommand{\pentaexchsixtop}{%
\begin{tikzpicture}[scale=0.6]
		\node (0) at (-5, 5) {\cirr};
		\node (1) at (-4, 5.75) {};
		\node (2) at (-3, 5) {};
		\node (4) at (-4, 3) {\cirr};
		\node (5) at (-3, 3.75) {};
		\node (6) at (-5, 3.75) {};
		\node (7) at (-2, 5) {};
		\node (8) at (-1, 5.75) {};
		\node (9) at (0, 5) {\cirr};
		\node (10) at (-1, 3) {\cirr};
		\node (11) at (0, 3.75) {};
		\node (12) at (-2, 3.75) {};
		\node (13) at (-3.5, 1.25) {};
		\node (14) at (-2.5, 2) {\cirr};
		\node (15) at (-1.5, 1.25) {};
		\node (16) at (-2.5, -0.75) {\cirr};
		\node (17) at (-1.5, 0) {};
		\node (18) at (-3.5, 0) {};
		\node (19) at (-5, -6.25) {\cirr};
		\node (20) at (-4, -5.5) {};
		\node (21) at (-3, -6.25) {};
		\node (22) at (-4, -8.25) {\cirr};
		\node (23) at (-3, -7.5) {};
		\node (24) at (-5, -7.5) {};
		\node (25) at (-2, -6.25) {};
		\node (26) at (-1, -5.5) {};
		\node (27) at (0, -6.25) {\cirr};
		\node (28) at (-1, -8.25) {\cirr};
		\node (29) at (0, -7.5) {};
		\node (30) at (-2, -7.5) {};
		\node (31) at (-3.5, -2.5) {\cirr};
		\node (32) at (-2.5, -1.75) {};
		\node (33) at (-1.5, -2.5) {\cirr};
		\node (34) at (-2.5, -4.5) {};
		\node (35) at (-1.5, -3.75) {};
		\node (36) at (-3.5, -3.75) {};
		\node (37) at (-3.5, 3.25) {};
		\node (38) at (-3, 1.75) {};
		\node (39) at (-1.5, 3.25) {};
		\node (40) at (-2, 1.75) {};
		\node (41) at (-2.5, -1) {};
		\node (42) at (-2.5, -1.5) {};
		\node (43) at (-3, -4.25) {};
		\node (44) at (-3.5, -5.75) {};
		\node (45) at (-2, -4.25) {};
		\node (46) at (-1.5, -5.75) {};
		\node (47) at (-5.25, 4.25) {};
		\node (48) at (-7, 3.5) {};
		\node (49) at (0.25, 4.25) {};
		\node (50) at (2, 3.5) {};
		\node (51) at (-5.25, -7) {};
		\node (52) at (-7, -8) {};
		\node (53) at (0.25, -7) {};
		\node (54) at (2, -8) {};
		\draw (0.center) to (1.center);
		\draw (1.center) to (2.center);
		\draw (2.center) to (5.center);
		\draw (5.center) to (4.center);
		\draw (4.center) to (6.center);
		\draw (6.center) to (0.center);
		\draw (0.center) to (4.center);
		\draw (4.center) to (1.center);
		\draw (4.center) to (2.center);
		\draw (7.center) to (8.center);
		\draw (8.center) to (9.center);
		\draw (9.center) to (11.center);
		\draw (11.center) to (10.center);
		\draw (10.center) to (12.center);
		\draw (12.center) to (7.center);
		\draw (7.center) to (10.center);
		\draw (10.center) to (8.center);
		\draw (10.center) to (9.center);
		\draw (13.center) to (14.center);
		\draw (14.center) to (15.center);
		\draw (15.center) to (17.center);
		\draw (17.center) to (16.center);
		\draw (16.center) to (18.center);
		\draw (18.center) to (13.center);
		\draw (13.center) to (16.center);
		\draw (16.center) to (14.center);
		\draw (16.center) to (15.center);
		\draw (19.center) to (20.center);
		\draw (20.center) to (21.center);
		\draw (21.center) to (23.center);
		\draw (23.center) to (22.center);
		\draw (22.center) to (24.center);
		\draw (24.center) to (19.center);
		\draw (19.center) to (22.center);
		\draw (22.center) to (21.center);
		\draw (25.center) to (26.center);
		\draw (26.center) to (27.center);
		\draw (27.center) to (29.center);
		\draw (29.center) to (28.center);
		\draw (28.center) to (30.center);
		\draw (30.center) to (25.center);
		\draw (25.center) to (28.center);
		\draw (28.center) to (27.center);
		\draw (31.center) to (32.center);
		\draw (32.center) to (33.center);
		\draw (33.center) to (35.center);
		\draw (35.center) to (34.center);
		\draw (34.center) to (36.center);
		\draw (36.center) to (31.center);
		\draw (31.center) to (34.center);
		\draw (34.center) to (33.center);
		\draw (31.center) to (33.center);
		\draw (19.center) to (21.center);
		\draw (25.center) to (27.center);
		\draw[red, style=dashed, very thick] (37.center) to (38.center);
		\draw[red, style=dashed, very thick] (39.center) to (40.center);
		\draw[very thick] (41.center) to (42.center);
		\draw[red, style=dashed, very thick] (43.center) to (44.center);
		\draw[red, style=dashed, very thick] (45.center) to (46.center);
		\draw[very thick] (48.center) to (47.center);
		\draw[very thick] (49.center) to (50.center);
		\draw[very thick] (51.center) to (52.center);
		\draw[very thick] (53.center) to (54.center);
		\filldraw[red] (-5, 5) circle (2pt);
		\filldraw[red] (-4, 3) circle (2pt);
		\filldraw[red] (0, 5) circle (2pt);
		\filldraw[red] (-1, 3) circle (2pt);
		\filldraw[red] (-2.5, 2) circle (2pt);
		\filldraw[red] (-2.5, -0.75) circle (2pt);
\end{tikzpicture}}

\newcommand{\rvline}{\hspace*{-\arraycolsep}\vline\hspace*{-\arraycolsep}}

\date{}
 \thanks{This work was partially supported by Simons Foundation collaboration grant for Mathematicians number 846970 and a London Mathematical Society Scheme 2 grant.}

\title[On super cluster algebras for Grassmannians]{On super cluster algebras based on super Pl\"ucker and super Ptolemy relations}
\author{Ekaterina~Shemyakova}
\address{Department of Mathematics,  University of Toledo, Toledo,  Ohio, USA}
\email{ekaterina.shemyakova@utoledo.edu}

\newtheorem{theorem}{Theorem}[section]

\newtheorem{lemma}{Lemma}[section]
\newtheorem{corollary}{Corollary}[section]
\theoremstyle{definition}

\newtheorem{example}{Example}[section]
\newtheorem{remark}{Remark}[section]

\begin{document}
\begin{abstract}
We study super cluster algebra structure arising in   examples provided by super Pl\"{u}cker and super Ptolemy relations.
We   develop   the super cluster structure of the    super Grassmannians  $\Gr_{2|0}(n|1)$ for arbitrary $n$, which was    indicated earlier in our  joint work with Th.~Voronov.
For the super Ptolemy relation for the decorated super Teichm\"{u}ller space of Penner-Zeitlin, we show how by a change of variables it can be transformed into the classical Ptolemy relation with the new even variables    decoupled from odd variables. We also   analyze   super Pl\"{u}cker relations for general super Grassmannians and obtain a new simple form of the relations for $\Gr_{r|1}(n|1)$. To this end, we establish properties of Berezinians of certain type matrices (which we call ``wrong'').
\end{abstract}

\maketitle
\tableofcontents

\section{Introduction}\label{sec.intro}

The paper is motivated by the problem of defining super cluster algebras, which has attracted a lot of attention in recent years. The current state of the art are the definitions  of Ovsienko~\cite{ovsienko-supercluster:2015} that were later modified into Ovsienko-Shapiro's definition~\cite{ovsienko-shapiro:2018}, as well as Li-Mixco-Ransingh-Srivastava's definition~\cite{srivastava:tosupercluster-2017}. Ovsienko-Shapiro
note the problem of the absence of any mutation of odd variables in their definition. The Li-Mixco-Ransingh-Srivastava's definition requires testing against expected properties and against  examples.

Such examples come by generalizing to the super case some of classical model examples. Known results in this direction comes from supergeometry. Such are Musiker-Ovenhouse-Zhang's~\cite{musiker2021expansion} super cluster algebra structure for the global coordinates on the decorated super Teichm\"uller space (based on the super Ptolemy relation of Penner-Zeitlin~\cite{pennerzeitlin2019decorated}), and Shemyakova-Voronov's~\cite{shemya:voronov2019:superplucker} cluster algebra construction based on their super Pl\"ucker relation for the homogeneous coordinates of the super Grassmannian $\Gr_{r|0}(n|1)$. The results of these two works have similarities and need  to be compared.

In this paper, we develop cluster construction for the super Grassmannian $\Gr_{2|0}(n|1)$    sketched in~\cite{shemya:voronov2019:superplucker}. We also make an important observation about the super Ptolemy relations of Penner-Zeitlin~\cite{pennerzeitlin2019decorated}   used by
Musiker-Ovenhouse-Zhang~\cite{musiker2021expansion} for their
super cluster structure. By a change of variables that we propose here, the super Ptolemy relation transforms into the classical Ptolemy relation in which the new even variables are decoupled from the odd variables.

Having in sight further application to super cluster algebras, we    analyze     super Pl\"ucker relations for general super Grassmannians. We obtain a new simple form of super Pl\"ucker relations for the case $\Gr_{r|1}(n|1)$ by introducing new odd variables, which have an additional advantage of being antisymmetric in all indices (which is what we need for simpler combinatorial properties).

Let us give some more detail.

Recall that the usual  Grassmannian $\Gr_k(n)$ is the space of $k$-dimensional
planes in an $n$-dimensional vector space. It is a smooth manifold, and there is a canonical isomorphism (``duality'') $\Gr_k(n) \cong \Gr_{n-k}(n)$. We consider the  super Grassmannian $\Gr_{r|s}(n|m)$, the space of $r|s$-dimensional planes in an $n|m$-dimensional super vector space. It is a supermanifold. Similarly with ordinary Grassmannians, there is a canonical isomorphism of the super Grassmannians $\Gr_{r|s}(n|m) \cong \Gr_{n-r|m-s}(n|m)$. One of the directions of this paper is exploration of the case of  $\Gr_{r|1}(n|1)$, which is dual to   $\Gr_{n-r|0}(n|1)$.

Every Grassmann manifold $\Gr_k(n)$ can be realized as a projective variety through the Pl\"ucker embedding, and the ring generated by its projective coordinates have the structure of a cluster algebra. For the case $\Gr_2(n)$, such cluster structure is finite and one can elegantly parameterize clusters by triangulations of $n$-gons. A cluster can mutate into another cluster and the variables of the new cluster can be expressed in terms of the variables of the original cluster by using Pl\"ucker relations. Geometrically,
this mutation corresponds to a ``flip'' of the diagonal in the corresponding quadrilateral formed by adjacent triangles of the triangulation.
In~\cite{shemya:voronov2019:superplucker}, we found a way to define the super Pl\"ucker embedding, where one of the starting difficulties is that it is well known that the super Grassmannian is not projective. We found a work-around by introducing a suitable weighted projected space with weights $+1,-1$.
We then obtained super Pl\"ucker relations for all   $\Gr_{r|s}(n|m)$, which  describe the image of the constructed embedding. In general, they are given by rational expressions because they are written in terms of Berezinians. Hence there remains a problem of ``getting rid of denominators'' and obtaining relations in polynomial form. In this paper, we partly solve this obtaining simple relations for the case of $\Gr_{r|1}(n|1)$. There is more to explore in this direction.

The structure of the paper is as follows. In Section~\ref{sec:prelim}, we recall the necessary notions concerning super exterior powers and Berezinians, and the construction of the super Pl\"{u}cker map of~\cite{shemya:voronov2019:superplucker}. In Section~\ref{sec:wrong_mat}, we deduce formulas for the Berezinian and inverse Berezinian for particular  super matrices which we call ``wrong''  matrices (those with a single row or column replaced by a vector of opposite parity). These formulas are used in Section~\ref{sec:dualplueck}, where we give a new form of super Pl\"ucker relations for $\Gr_{r|1}(n|1)$ (dual to $\Gr_{n-r|0}(n|1)$). In Section~\ref{sec:supercluster}, we introduce  the super cluster structure of the homogeneous coordinate ring of the super Grassmannian $\Gr_{2|0}(n|1)$. In particular, we define super clusters, introduce ``even'' and ``odd'' mutations, and describe the exchange graph. Super clusters here are parameterized by certain decorated triangulations of the $n$-gon and mutations are given by geometric rules.  Finally, in  Section~\ref{sec:superptol}, we consider the super Ptolemy relation of Penner and Zeitlin and show how by a change of variables it can be transformed into the classical Ptolemy relation.

\section{Preliminaries: Berezinians, super exterior powers,  and super Pl\"{u}cker map}
\label{sec:prelim}

In this section, we give some background material and  review our results from~\cite{shemya:voronov2019:superplucker}.

\subsection{Supermatrices and Berezinians.}
Below are basic notions that we use throughout the paper. A (super) (vector) space $V$ is a $\ZZ$-graded vector space $V=V_0 \oplus V_1$ over numbers, say $\RR$, where elements of $V_0$ are called \emph{even} and elements of $V_1$ are called \emph{odd}.
One can always choose a homogeneous basis containing even basis vectors $e_1, \dots, e_n$ and odd basis vectors
$e_{\widehat{1}}, \dots, e_{\widehat{m}}$.
Super algebras are defined in a natural way so that    $\ZZ$-grading is respected (parities are added under multiplication) and commutativity is understood in the graded sense.  To allow non-trivial coefficients for vectors, one  can construct from a vector space $V$, a free module over $A$  as $A \otimes V$, for some auxiliary commutative super algebra $A$.


Matrices of linear operators acting on vector spaces or free modules over a commutative superalgebra have their rows and
columns labeled with  parities corresponding to the parities of the basis vectors --- \emph{supermatrices}.
It is always possible to (uniquely) transform any matrix into \emph{standard format} where all even rows go first and all odd rows go after them, and the same for columns:
\begin{equation}
 \begin{pmatrix}
      A_{00} & \rvline & A_{01} \\
      \hline
      A_{10} & \rvline & A_{11}
    \end{pmatrix} \ .                                                                                                                                  \end{equation}
A matrix is \emph{even} (\emph{odd}) if it has even (odd) entries in $A_{00}$ and $A_{11}$, and odd (even) entries in $A_{01}$ and $A_{10}$.

The super analog of determinant, called \emph{Berezinian},   is a essentially unique, up to taking powers, multiplicative function $\Ber A$ on the space of even invertible matrices.
As a function of rows (columns) $\Ber A$ is homogeneous of degree $+1$ in each even row (column);
of degree $-1$ in each odd row (column); unchanged under elementary row (and elementary column) transformations; and $\Ber E=1$ for an identity matrix $E$. The explicit formula is
\begin{equation}\label{eq.berexpl}
    \Ber A=\frac{\det(A_{00}-A_{01}A_{11}^{-1}A_{10})}{\det A_{11}}=\frac{\det A_{00}}{\det(A_{11}-A_{10}A_{00}^{-1}A_{01})}\,.
\end{equation}
%
The \emph{parity reversion} of a matrix is the change of its format so that all the parities of the rows and columns are replaced by the opposite (nothing happens with the entries themselves). Notation: $A^{\Pi}$.
Following~\cite{rabin:duke}, the \emph{inverse Berezinian} $\Ber^*A$ is defined as
\begin{equation}\label{eq.berpi}
    \Ber^*A:= \Ber A^{\Pi} \ ,
\end{equation}
and so
\begin{equation}\label{eq.berinv}
    \Ber^* A=\frac{\det(A_{11}-A_{10}A_{00}^{-1}A_{01})}{\det A_{00}}=\frac{\det A_{11}}{\det(A_{00}-A_{01}A_{11}^{-1}A_{10})}\,.
\end{equation}
$\Ber A$ is defined if and only if $A_{11}$ is invertible, while $\Ber^*A$ is defined if and only if $A_{00}$ is invertible. If both $A_{00}$ and $A_{11}$ are invertible, then $\Ber A$ and $\Ber^*A$ both make sense and are mutually reciprocal. $\Ber A$, as a function of rows (columns), is multilinear in all even arguments; $\Ber^*A$ is multilinear in all odd arguments. Note that rows are multiplied by scalars on the left and columns, on the right.

\subsection{Super exterior powers.}
There are two cases of super Grassmannians $Gr_{r|s}(n|m)$: that of completely even $r$-planes  (i.e.,  $r|0$-)   in a super  $n|m$-space, and the general case of $r|s$-planes.
The constructions of super exterior powers matching each of these cases are  as follows.

Consider first $\L^k(V)=\L^{k|0}(V)$, which 
can be constructed algebraically, by using the quotient of the tensor algebra $T(V)$ by the
anti-commutativity relations
$\ep_{i}\otimes \ep_{j} +\ep_{j}\otimes \ep_{i}$, $\ep_{i}\otimes \ep_{\widehat{a}}+\ep_{\widehat{a}}  \otimes \ep_{i}$,
$\ep_{\widehat{a}} \otimes \ep_{\widehat{b}} -   \ep_{\widehat{b}} \otimes \ep_{\widehat{a}}$. Multivectors $T\in \L^k(V)$ have the form
\begin{equation}\label{eq.Twedge}
    T=T^{a_1\ldots a_k}\ep_{a_1}\wed \ldots \wed \ep_{a_k}\,,
\end{equation}
with the (super)antisymmetry condition $T^{a_1\ldots a_ia_{i+1} \ldots a_k}=-(-1)^{\at_i\at_{i+1}}T^{a_1\ldots a_{i+1}a_i \ldots a_k}$.

The general case of $Gr_{r|s}(n|m)$, $s \neq 0$ requires a more complicated construction of super exterior powers  denoted $\L^{r|s}(V)$~\cite{shemya:voronov2019:superplucker} --- an adaptation of the ``Voronov--Zorich'' construction~\cite{tv:pdf,tv:cohom-1988,tv:git,tv:dual,tv:cartan1}   originally proposed for  forms. For a superspace  $V$, $\dim V=n|m$, a
\emph{multivector of degree $r|s$} is a function $F(p)=F(p^1,\ldots,p^{r+s})$ of $r$ even and $s$ odd covectors $p^i\in V^*$ satisfying two conditions:
\begin{align}\label{eq.bercond}
    &F(p\cdot g)  =F(p)\cdot\Ber g\,, \\
     & \dder{F}{{p_a{}^i}}{{p_b{}^j}}  + (-1)^{\itt\,\jt +\at(\itt+\jt)} \dder{F}{p_a{}^j}{p_b{}^i} = 0\,, \label{eq.fund}
\end{align}
for all $g\in GL(r|s)$ and all combinations of indices $i,j,a,b$.
(Here $r\geq 0$, but there is also an extension for negative $r$.) We use  left coordinates of covectors  writing them as columns, so $p$ as the argument of $F(p)$ can be represented
by an (even) $n|m \times r|s$  matrix.
Condition~\eqref{eq.bercond}
implies that $F(p)$ is invariant under     elementary transformations (adding to a column a multiple of another column) and  it is homogeneous of degree $+1$ with respect to each of  $r$ even columns and homogeneous of degree $-1$ in each of  $s$  odd columns. Equations~\ref{eq.fund} are a sophisticated replacement of multilinearity and skew-symmetry of ordinary multivectors. When $s=0$, constructed in such a way $\L^{r|0}(V)$ can be identified with $\L^r(V)$ constructed as above. There is a deep   connection between the theory of super exterior powers and integral geometry in the sense of Gelfand--Gindikin--Graev. Moreover, equations~\eqref{eq.fund} can be seen as a ``super generalization'' of Gelfand's general hypergeometric system. This relation was   used in~\cite{tv:cohom-1988}, but  there is definitely  more to explore.

There is an explicit simple formula
for a  non-linear  \emph{wedge product of even and odd vectors} taking values in $\L^{r|s}(V)$.
Let $\up_1,\ldots,\up_r,\up_{r+1},\ldots,\up_{r+s}$ be a sequence of $r$ even and $s$ odd linearly independent vectors. Define a function $[\up_1,\ldots,\up_r|\up_{r+1},\ldots,\up_{r+s}]$ by
\begin{equation}\label{eq.nonlinwed}
    [\up_1,\ldots,\up_r|\up_{r+1},\ldots,\up_{r+s}](p):=\Ber (\up_i{}^ap_a{}^j)\,.
\end{equation}
This function  satisfies~\eqref{eq.bercond} and~\eqref{eq.fund}  and so is an element  of $\L^{r|s}(V)$. When $s=0$, \eqref{eq.nonlinwed} would coincide with $\up_1\wed \ldots \wed \up_r$ up to a factor.
In the $s \neq 0$ case, it is not known if such simple multivectors
span the whole $\L^{r|s}(V)$. This space can be infinite-dimensional and its explicit description is not yet known.

\subsection{Super Pl\"{u}cker map.}
The (super) Pl\"ucker embedding and the (super) Pl\"ucker coordinates in the  Grassmannian $\Gr_{r|s}(n|m)$ were constructed in~\cite{shemya:voronov2019:superplucker} as follows. Fix  a basis in the superspace $V$ as above. Let vectors in $V$ be described by their left coordinates, which we write as rows. A basis $\up_i$ spanning a plane $L$ will be represented by an $r|s\times n|m$ matrix $U$, which we consider as the matrix of homogeneous coordinates of a point $L$ of the super Grassmannian $\Gr_{r|s}(V)$. Homogeneity is with respect to the action of the supergroup $\GL(r|s)$, i.e. $U\sim gU$ for all $g\in GL(r|s)$.

In the case of $\Gr_{r|0}(n|m)$, there is a well-defined map, which can be proved to be an embedding:
\begin{equation} 
   Gr_{r|0}(n|m)  \to P(\L^r(V))\,, \quad  
    L  \mapsto [\up_1\wed \ldots \wed \up_r]\,.
\end{equation}

In the general case of $\Gr_{r|s}(n|m)=\Gr_{r|s}(V)$, such a map will not be well-defined (the image of $L$ would depend on a choice of a basis in the plane). Instead
one
maps the super Grassmannian $\Gr_{r|s}(V)$ into a carefully created weighed projective space using~\eqref{eq.nonlinwed} and its parity-dual. Covectors $p\in V^*$ can be represented by their right coordinates written as columns, so that
an array of $r|s$ covectors will be represented by an $n|m\times r|s$ matrix $P$. For an $L \in \Gr_{r|s}(n|m)$, given by an even $r|s\times n|m$ matrix $U$ of rank $r|s$,  define
functions $\pl(U)$ and $\pl^*(U)$ on even $n|m\times r|s$ matrices  $P$ of rank $r|s$,
 \begin{align}
    \pl &:  Gr_{r|s}(V) \rightarrow \Lambda^{r|s}(V)  \, , \
    \label{eq.plup}
   \pl(U): P \mapsto \Ber (UP)\,, \\
     \pl^* &:  Gr_{r|s}(V) \rightarrow \Lambda^{s|r}(\Pi V)  \, , \
     \label{eq.plup*}
   \pl^*(U):  P \mapsto \Ber^* (UP)\,.
 \end{align}
Then we consider
\begin{equation}
 \Pl= [\pl  \oplus \pl^*] :  \ Gr_{r|s}(n|m) \rightarrow \pfinth \ ,
\end{equation}
where $\pfinth$ is the weighted projectivization of the
space with coordinates
given by functions
$\pl(U)$ and $\pl^*(U)$ evaluated at even and
\emph{wrong even matrices} $P$   made of basis covectors (see more on wrong matrices in Sec.~\ref{sec:wrong_mat}):
\begin{align}\nonumber
    T^{a_1\ldots a_r|\hmu_1\ldots\hmu_s}&:=\pl(U)(e^{a_1},\ldots,e^{a_r}|e^{\hmu_1},\ldots,e^{\hmu_s}) \\
    &=\Ber U^{a_1\ldots a_r|\hmu_1\ldots\hmu_s}  \, ,  \label{eq.ua1mus}\\
    \nonumber
    T^{*a_1\ldots a_r|\hmu_1\ldots\hmu_s}&:=\pl^*(U)(e^{a_1},\ldots,e^{a_r}|e^{\hmu_1},\ldots,e^{\hmu_s})   \\
    & =\Ber^* U^{a_1\ldots a_r|\hmu_1\ldots\hmu_s} \,, \label{eq.usta1mus}\\
    \nonumber
    T^{a_1\ldots a_{j-1}\hnu a_{j+1}\ldots a_r|\hmu_1\ldots\hmu_s}&:=\pl(U)(e^{a_1},\ldots,e^{a_{j-1}},e^{\hnu}, e^{a_{j+1}},\ldots, e^{a_r}|e^{\hmu_1},\ldots,e^{\hmu_s})  \\
    &= \Ber U^{a_1\ldots a_{j-1}\hnu a_{j+1}\ldots a_r|\hmu_1\ldots\hmu_s}  \,, \label{eq.ua1numus} \\
    \nonumber
    T^{*a_1\ldots a_r|\hmu_1\ldots \hmu_{\be-1}b \hmu_{\be+1}\ldots \hmu_s}&:=\pl^*(U)(e^{a_1},\ldots,e^{a_r}|e^{\hmu_1},\ldots,e^{\hmu_{\be-1}},e^b,e^{\hmu_{\be+1}},\ldots,e^{\hmu_s}) \\
    &= \Ber^* U^{a_1\ldots a_r|\hmu_1\ldots \hmu_{\be-1}b \hmu_{\be+1}\ldots \hmu_s} \,. \label{eq.usta1bmus}
\end{align}
%
Here $U^{a_1\ldots a_r|\hmu_1\ldots\hmu_s}$ is the   submatrix of the matrix $U$ obtained by choosing $r$ even and $s$ odd columns with the indicated indices. The so defined variables are called the \textit{(super) Pl\"{u}cker coordinates} for $\Gr_{r|s}(V)$. As proved in~\cite{shemya:voronov2019:superplucker}, $\Pl$ is an embedding. 
\begin{remark}
Observe that Pl\"ucker coordinates $T^{a_1\ldots a_r|\hmu_1\ldots\hmu_s}$ and $T^{*a_1\ldots a_r|\hmu_1\ldots\hmu_s}$ are even, while $T^{a_1\ldots a_{j-1}\hnu a_{j+1}\ldots a_r|\hmu_1\ldots\hmu_s}$  and
$T^{*a_1\ldots a_r|\hmu_1\ldots \hmu_{\be-1}b \hmu_{\be+1}\ldots \hmu_s}$ are odd (since they are obtained by computing $\Ber$ and $\Ber^*$ of wrong matrices). Also, $T^{*a_1\ldots a_r|\hmu_1\ldots\hmu_s}=(T^{a_1\ldots a_r|\hmu_1\ldots\hmu_s})^{-1}$.
 \end{remark}



\section{Properties of $\Ber$ and $\Ber^*$ of ``wrong'' matrices}
\label{sec:wrong_mat}


\subsection{``Wrong'' matrices and Berezinians.}
By linearity, one generalizes Berezinian
to a ``wrong'' in the sense of parity matrix $A$,  which is obtained from an even matrix by placing an odd vector (such a vector has odd entries in the even positions and even entries in the odd positions)
into an even row (or column). Similarly, $\Ber^* A$ is generalized for a ``wrong'' matrix where at an odd row or an odd column an even vector is placed. 

Wrong matrices appear naturally in the row and column expansion formulas for Berezinian~\cite{2015:super,shemya:voronov2016:berezinians} and in  super  Pl\"ucker coordinates~\cite{shemya:voronov2019:superplucker}.
Practically, $\Ber$ and $\Ber^*$ of ``wrong'' matrices are calculated by the same explicit formulas as for the case of even matrices. 
Many properties of  $\Ber$ and $\Ber^*$   are preserved by this extension. In addition, wrong matrices have new properties which in what follows leads to a crucial simplification.

The following lemma contains a few properties that follow immediately from the explicit formulas  and properties of Berezinians and determinants.

\begin{lemma}
 \begin{enumerate}
  \item Wrong matrices are neither even nor odd.
  \item $\Ber$ and $\Ber^*$ take odd values on wrong matrices.
  \item
Extended to wrong matrices,
  $\Ber A$ and $\Ber^* A$ are still invariant under elementary transformations provided multiples of ``correct'' vectors are added to the ``wrong'' vector, not the other way round.
  \item $\Ber$ and $\Ber^*$ are antisymmetric (in the usual non-super sense) within each group of even and odd rows (columns).
 \end{enumerate}
\end{lemma}

\begin{example}  If one does not follow the above rules, it is possible to run into a contradiction. Indeed, let $x$ be even, $\xi$ be odd. Consider a wrong matrix of dimension $1|1$ (assuming standard format),
\begin{equation*}A =\begin{pmatrix}
  \xi   & \rvline & x \\
\hline
\xi   & \rvline & x
\end{pmatrix}
\end{equation*}
where the odd vector $(\xi, x)$ stands in an even row. Using the explicit formula, we have $\Ber A =(\xi-xx^{-1}\xi)/x=0$.
Subtracting the second row from the first, and then by computing Berezinian we get the same answer; while subtracting the first from the second, we would get a matrix for which $\Ber$ is not defined. The latter transformation is prohibited.
\end{example}

\subsection{Particular case: $r|1\times r|1$-matrices.}
\begin{theorem} \label{thm:my_trick}
For an $r|1$-square wrong  matrix $A$
obtained from an even matrix by putting some even vector in place of its only odd column, one has
\begin{equation} \label{eq:Ber_inv_property}
 \Ber^* A  = \frac{\det A}{{\det}^2 A_{00}} \ .
\end{equation}
For an $1|r$-square wrong matrix $B$
obtained from an even matrix with an odd vector   placed in its only even column, one has
\begin{equation} \label{eq:Ber_property}
 \Ber B  = \frac{\det B}{{\det}^2 B_{11}} \ .
\end{equation}
\end{theorem}
Here $\det A$ of a super matrix $A$ as above is computed by forgetting all the information about the parities of the rows/columns and then evaluating the determinant as usual. Note that in this particular case there is no problem with commutativity of matrix elements since every monomial is the standard expansion of the determinant will contain only one odd factor. (The resulting determinant will have odd value.)
\begin{proof}
Formulas~\eqref{eq:Ber_inv_property} and~\eqref{eq:Ber_property}  are proved analogously. Let us prove, for example,~\eqref{eq:Ber_inv_property}.
Consider a wrong $r|1$-square matrix $A$. Consider
\begin{equation}
 \Ber^* A =
\Ber^*
\begin{pmatrix}
A_{00}
 & \rvline &  \begin{matrix} u_1^c  \\ \vdots  \\ u_r^c  \end{matrix} \\
\hline
\begin{matrix} u_{\hat{1}}^1  & \ldots  & u_{\hat{1}}^r  \end{matrix}& \rvline & u_{\hat{1}}^c
\end{pmatrix}
=
\frac{u_{\hat{1}}^c \det A_{00} - \begin{pmatrix} u_{\hat{1}}^1  & \ldots  & u_{\hat{1}}^r  \end{pmatrix}\adj A_{00} \begin{pmatrix} u_1^c  & \ldots  & u_r^c  \end{pmatrix}^T}{\det^2 A_{00}} \ .
\label{eq_pr1}
\end{equation}
Here ``hatted'' indices are odd and not ``hatted'' are even. An element with two indices that are both ``hatted'' is even. Block $A_{00}$ has only even entries and $A_{00}^{-1} = \adj A_{00}/\det A_{00}$, and $(\adj A_{00})_{ij}=(-1)^{i+j} (A_{00})_{ji}$, where $\adj A_{00}$ is the adjugate to $A_{00}$ and $(A_{00})_{ji}$ is the determinant of the $(r-1) \times (r-1)$ matrix that results from deleting row $j$ and column $i$ of $A$.

First, lets us compute the $i$-th component of column vector $\adj A_{00} \begin{pmatrix} u_1^c  & \ldots  & u_r^c  \end{pmatrix}^T$. We have
 $\sum_{j=1}^r (-1)^{i+j} (A_{00})_{ji} u_j^c$. This can be understood as the expansion along the $i$-th row of the determinant of the matrix obtained from $A_{00}$ by replacing its $i$-th column by $\begin{pmatrix} u_1^c  & \ldots  & u_r^c  \end{pmatrix}^T$, which is actually a maximal minor of matrix $A$,
\begin{equation} \label{eq_pr2}
 \det  A_{1\dots r}^{1\dots i-1 \, c \, i+1 \dots r} = (-1)^{r-i+1} \det  A_{1\dots r}^{1\dots i-1 \, i+1 \dots r \, c} \ .
\end{equation}
Plugging~\eqref{eq_pr2} into~\eqref{eq_pr1} we have
\begin{equation}
 \Ber^* A \, {\det}^2 A_{00} =
u_{\hat{1}}^c \det A_{00} + \sum_{i=1}^r (-1)^{r+i} u_{\hat{1}}^i  \det  A_{1 \dots r}^{1\dots i-1 \, i+1 \dots r \, c} \ .
\end{equation}
The last expression is exactly the expansion of the determinant of matrix $A$ along its last row $\hat{1}$. In this discussion $A_{i_1\dots i_r}^{j_1\dots j_r}$ stands for the matrix obtained from $A$ by keeping rows $i_1, \dots, i_r$ and columns $j_1, \dots, j_r$ (and deleting the remaining one row and one column). Note that in matrix $A$, the columns are numbered $1, \dots, r, c$ and the rows are numbered $1, \dots, r, \hat{1}$.
\end{proof}

\begin{corollary} \label{cor:my_trick}
$\Ber^* A \cdot {\det}^2 A_{00}$ is  antisymmetric function of the columns (rows) of $A$.
\end{corollary}

\subsection{New super Pl\"{u}cker coordinates for the super Grassmannian $\Gr_{r|1}(n|1)$.}

Corollary~\ref{cor:my_trick}   allows  us to
re-write the super Pl\"ucker relations for $\Gr_{r|1}(n|1)$ in new odd  homogeneous coordinates that are   antisymmetric in all indices. 
The only problem   is that   Corollary~\ref{cor:my_trick} gives this directly only for the ``normalized case'' (see more in Example~\ref{ex:1}). This   shortfall will be overcome in Theorem \ref{thm:trick_general_case}.

\begin{example}
\label{ex:1}
Let us consider a special case where $r=2$, that is $2|1$-planes $L$ in the $n|1$-space.  Consider,  as we did above, the  matrix $U$, the matrix of homogeneous coordinates of a plane $L$.
As the rank of this matrix is $2|1$, 
there is a basis for $L$ where the last column becomes the standard basis vector:  
 \begin{equation}U =
 \begin{pmatrix}
    u_1^1 & \dots &u_1^n &\rvline &0  \\
    u_2^1 & \dots &u_2^n &\rvline &0  \\
   \hline
    u_{\widehat{1}}^1 & \dots &  u_{\widehat{1}}^n &\rvline & 1
\end{pmatrix} \ .
\end{equation}
This normalization implies
\begin{equation} \label{eq:simpl}
T^{ab|\hat{1}}=
\det A_{00} \, ,
\end{equation}
where $A_{00}$ is the even-even block of matrix $U$. (Recall~\eqref{eq.ua1mus} for the definition of $T^{ab|\hat{1}}$.)
In this case the  super Pl\"ucker coordinates~\eqref{eq.ua1mus},\eqref{eq.usta1mus},\eqref{eq.usta1bmus}
are even $T^{ab|\hat{1}}$, $T^{*ab|\hat{1}}$, and odd
$T^{ab|c}$, while odd coordinates~\eqref{eq.ua1numus} are always zero.  Here the index $\hat{1}$ corresponds to the only odd column in
the matrix $U$.

There is an inconvenience that
$T^{ab|c}=\Ber^* U^{ab|c}$ is antisymmetric in just $a$ and $b$. Theorem~\eqref{thm:my_trick} together with~\eqref{eq:simpl} allow  us to introduce new homogeneous odd coordinates:
\begin{equation}
 \t^{abc} := T^{ab|c} (T^{ab|\hat{1}})^2
\end{equation}
that are anti-symmetric in all indices.
\end{example}

This cannot be lifted in any immediate way to the case of general (not normalized) matrix $U$ since it involves non-maximal minors. Below is a generalization of Theorem~\eqref{thm:my_trick} which will serve the general case.

\begin{theorem} \label{thm:trick_general_case}
For $L \in \Gr_{r|1}(n|1)$ consider matrix $U$
whose rows are basis vectors of $L$.
Consider one odd homogeneous Pl\"ucker coordinate~\eqref{eq.usta1bmus}, $T^{*\aund|c}$, and one even homogeneous
coordinate~\eqref{eq.ua1mus}, $T^{\aund|\hat{1}}$, where $\aund=a_1 \dots a_r$. Then
\begin{equation} \label{eq:my_equality}
 T^{*\aund|c} \left( T^{\aund|\hat{1}} \right)^2 = \frac{\det U^{\aund c} \left(\det U^\aund_{1 \dots r}\right)^2}{\left( \det U^{\aund\hat{1}}\right)^2} \ .
\end{equation}
Here $U^{\aund c}$ and $U^{\aund\hat{1}}$ stand for the maximal submatrices of $U$; $U^\aund_{1 \dots r}$ stands for the $r \times r$ submatrix of $U$. (The upper indices indicate the selected columns. The lower indices indicate the selected rows.)
\end{theorem}

\begin{corollary} \label{cor:my_trick_general_case}
 \begin{equation}
 \t^{\aund c}:=T^{*\aund|c} \left( T^{\aund|\hat{1}} \right)^2
 \end{equation}
is anti-symmetric in all its indices $a_1,\ldots,a_r,c$.
\end{corollary}

\begin{proof}[Proof of Theorem~\ref{thm:trick_general_case}]
Consider
\begin{equation}
U =
 \begin{pmatrix}
   u_1^1 &  \dots &u_1^r &\rvline &u_1^{\hat{1}}  \\
      \vdots  &  \dots &\vdots  &\rvline &\vdots   \\
   u_r^1 &  \dots &u_r^r &\rvline &u_r^{\hat{1}}  \\
   \hline
   u_{\hat{1}}^1 &  \dots &u_{\hat{r}}^r &\rvline &u_{\hat{1}}^{\hat{1}}  \\
\end{pmatrix}
\end{equation}
and denote as before $A_{00}=U^\aund_{1 \dots r}$. Then
by using Theorem~\ref{thm:my_trick}, we have
\begin{multline}
 T^{*\aund|c} \left( T^{\aund|\hat{1}} \right)^2 =
 \frac{\det U^{\aund c}}{\det^2 A_{00}}
 \cdot
 \Ber^2
\begin{pmatrix}
A_{00}
 & \rvline &  \begin{matrix} u_1^{\hat{1}}  \\ \vdots  \\ u_r^{\hat{1}} \end{matrix} \\
\hline
\begin{matrix} u_{\hat{1}}^1  & \ldots  & u_{\hat{1}}^r  \end{matrix}& \rvline & u_{\hat{1}}^{\hat{1}}
\end{pmatrix}
\\
=
\frac{\det U^{\aund c}}{\det^2 A_{00}}
 \cdot
\frac{\det^2 A_{00}}{\left(u_{\hat{1}}^{\hat{1}} -
\begin{pmatrix} u_{\hat{1}}^1  & \ldots  & u_{\hat{1}}^r  \end{pmatrix} \frac{\adj(A_{00})}{\det A_{00}} \begin{pmatrix} u_1^{\hat{1}}  \\ \vdots  \\ u_r^{\hat{1}} \end{pmatrix}\right)^2} \\ =
\frac{\det U^{\aund c} \det^2 A_{00}}{\left(\det A_{00} u_{\hat{1}}^{\hat{1}} -
\begin{pmatrix} u_{\hat{1}}^1  & \ldots  & u_{\hat{1}}^r  \end{pmatrix} \adj(A_{00}) \begin{pmatrix} u_1^{\hat{1}}  & \ldots & u_r^{\hat{1}} \end{pmatrix}^T\right)^2}
\, . \label{eq:t}
\end{multline}
Here $\adj(A_00)$ is the adjugate of $A_{00}$.
Compute the $i$th element of the column vector $ \adj(A_{00}) \begin{pmatrix} u_1^{\hat{1}} & \ldots & u_r^{\hat{1}} \end{pmatrix}^T$:
$$\sum_{j=1}^r
\cofactor_{ji} (A_{00}) u_j^{\hat{1}}=\det A_{00}^{a_1 \dots a_{i-1} \hat{1} a_{i+1} \dots a_{r}}=(-1)^{r-i} \det A_{00}^{a_1 \dots a_{i-1} a_{i+1} \dots a_{r} \hat{1}} \, .$$
In the latter we first used that the expression involving cofactors is the expansion over column $\hat{1}$ of the determinant of matrix $A_{00}^{a_1 \dots a_{i-1} \hat{1} a_{i+1} \dots a_{r}}$ which is the matrix obtained from $A_{00}$ by putting the  column vector $ \begin{pmatrix} u_1^{\hat{1}} & \ldots & u_r^{\hat{1}} \end{pmatrix}^T$ in place of its $i$th column.  Secondly, we used anti-symmetry of  determinant to shift the column $\hat{1}$ to the right.

Thus the denominator of~\eqref{eq:t} before squaring equals
\begin{equation*}
  \det A_{00} u_{\hat{1}}^{\hat{1}} + \sum_{i=1}^{r} (-1)^{r+i+1} u_{\hat{1}}^i A_{00}^{a_1 \dots a_{i-1} a_{i+1} \dots a_{r} \hat{1}} =\det U^{\aund \hat{1}}\,,
\end{equation*}
where the latter equality comes from noticing the row expansion over row $\hat{1}$ of the determinant of $U^{\aund \hat{1}}$. Hence we get~\eqref{eq:my_equality}. (Note that we use some ordering of odd variables to define this determinant.)
\end{proof}

\begin{remark}
  Under a change of basis in $L$, $U\mapsto gU$, the variables $\t^{\aund c}$ transform as $\t^{\aund c}\mapsto (\Ber g)\,\t^{\aund c}$.
\end{remark}


\section{Super Pl\"ucker relations for $\Gr_{r|1}(n|1)$, the dual case for $\Gr_{r|0}(n|1)$}

\label{sec:dualplueck}

\subsection{Recollection of super Pl\"ucker relations for even planes}
Consider first the ``algebraic case''. Here $T^{a_1 \dots a_r}$ is (super) antisymmetric in indices.


\begin{theorem}[\cite{shemya:voronov2019:superplucker}]
\label{thm.splueckcompon}
An even multivector $T\in \L^r(V)$ in an $n|m$-dimensional superspace $V$ is simple and thus defines an $r$-plane $L\subset V$ if and only if it is non-degenerate and its components~\eqref{eq.Twedge} satisfy the relations
\begin{multline}\label{eq.superpluecksimpl}
    T^{a_1\ldots a_{r-1}b}T^{c_1\ldots c_k}\,(-1)^{\bt(\at_1+\ldots+\at_{r-1}+\ct_1+\ldots+\ct_{r})}
    \\
    =\sum_{j=1}^{r}  T^{a_1\ldots a_{k-1}c_j}T^{c_1\ldots c_{j-1}bc_{j+1}\ldots c_r}\,
    (-1)^{\bt(\ct_1+\ldots+\ct_{j-1}) + \ct_j(\at_1+\ldots+\at_{r-1}+\ct_{j+1}+\ldots+\ct_{r})}
\end{multline}
for all combinations of indices $a_1$, \ldots, $a_{r-1}$, $b$, $c_1$, \ldots, $c_r$.
\end{theorem}

(``Non-degenerate'' here means that one of the components $T^{a_1 \dots a_r}$ where all indices are even, is invertible.)

Restricting these relations to the case of $2$-planes, we get the following relations.

\begin{theorem}[\cite{shemya:voronov2019:superplucker}]
\label{thm.splueckcompon}
An even bivector $T\in \L^2(V)$ in an $n|m$-dimensional superspace $V$ is simple and thus defines a $2$-plane $L\subset V$
if and only if it is non-degenerate and its components satisfy the relations
 \begin{equation}\label{eq.plueckfor2}
    T^{ab}T^{cd}(-1)^{\bt(\at+\ct+\dt)} =
    T^{ac}T^{bd}(-1)^{\ct(\at+ \dt)} +
    T^{ad}T^{cb}(-1)^{\bt\ct + \at\dt}\,.
 \end{equation}
for all combinations of indices $a,b,c,d$. Each of the indices can be even or odd   and also the indices can repeat.
\end{theorem}

\begin{example}
 For the particular case of $\Gr_{2|0}(n|1)$, relations~\eqref{eq.plueckfor2} become
\begin{align}
 T^{ac} T^{bd} &= T^{ab} T^{cd} + T^{ad} T^{bc} \ , \label{eq:first_mutation_rule} \\
\label{eq:second_mutation_rule}
    T^{ab}\t^{c} &= T^{ac}\t^{b}  + T^{cb}\t^{a}\,,
\end{align}
and
\begin{align}
    T^{ab} T^{\hat{1} \hat{1}} &= - 2\t^{a}\t^{b}\,, \\
\label{eq.plueck.thetas}
    \t^{a} T^{\hat{1} \hat{1}} &= 0\,, \\
\label{eq.plueck.skw}
    \left(T^{\hat{1} \hat{1}}\right)^2 &= 0\,,
\end{align}
where $\t^a=T^{a\hat{1}}$, and the indices $a,b,c,d$ are even and run over  $1,\ldots,n$.
As every chart in $\Gr_{2|0}(n|1)$ has some $T^{ab}$ invertible,  equality~\eqref{eq.plueck.thetas} can be used to express the even nilpotent $T^{\hat{1} \hat{1}}$ in terms of the remaining variables.  Hence we are left with only two kinds of relations:~\eqref{eq:first_mutation_rule} and~\eqref{eq:second_mutation_rule}.
\end{example}

So we can see that some of the relations can be reduced. This idea  carries over to the general case of $r$-planes, and we have the following result.

\begin{theorem}[\cite{shemya:voronov2019:superplucker}]
\label{thm.essplueckr0}
On $\pessth$, the super Pl\"{u}cker relations take the form
\begin{align}\label{eq.esssplueckev}
    T^{a_1\ldots a_r}T^{b_1\ldots b_r}&=\sum_{j=1}^{r} T^{b_ja_2\ldots a_r}T^{b_1\ldots b_{j-1}a_1b_{j+1}\ldots b_r} &\ \text{\emph{(even)}}
    \\
    \label{eq.esssplueckod}
    T^{a_1\ldots a_r}T^{b_1\ldots b_{r-1}\hmu}&=\sum_{j=1}^{r-1} T^{b_ja_2\ldots a_r}T^{b_1\ldots b_{j-1}a_1b_{j+1}\ldots b_{r-1}\hmu}+ T^{\hmu a_2\ldots a_r}T^{b_1\ldots b_{r-1}a_1} &\ \text{\emph{(odd)}}
\end{align}
for all combinations of even indices $a_1$, \ldots, $a_{r}$,  $b_1$, \ldots, $b_r$ and odd index $\hmu$.
\end{theorem}

Here $\pessth$ is the projective superspace with the variables $T^{a_1\ldots a_r}, T^{b_1\ldots b_{r-1}\hmu}$ as homogeneous coordinates. 

\subsection{Recollection of super Pl\"ucker relations for the general case}

Consider the super Pl\"ucker relations obtained in~\cite{shemya:voronov2019:superplucker} for $r|s$-planes in $n|m$-space.
\begin{theorem}[\cite{shemya:voronov2019:superplucker}]
\label{thm:general_plu}
The super Pl\"{u}cker relations for the homogeneous coordinates on   $\pfinth$, are:
\begin{align}
\label{eq.relrs0}
&T^{\aund|\hmund}\,T^{*\aund|\hmund}=1\,,
\\
\label{eq.relrs1}
    &\Ber
    \left(\begin{array}{ccc}
        T^{a_1\ldots a_{i-1}b_ja_{i+1}\ldots a_r|\hmund} & \vline & T^{a_1\ldots a_{i-1}\hnu_{\be}a_{i+1}\ldots a_r|\hmund} \\
                                 \hline \vphantom{\int_A^B}
         T^{*\aund|\hmu_1\ldots\hmu_{\al-1}b_j\hmu_{\al+1}\ldots \hmu_s} & \vline & T^{*\aund|\hmu_1\ldots\hmu_{\al-1}\hnu_{\be}\hmu_{\al+1}\ldots \hmu_s} \vphantom{\int_A^B}\\
        \end{array}\right)_{{i,j=1\ldots r,}\atop{\al,\be=1\ldots s}}
        =(T^{\aund|\hmund})^{r+s-1} T^{\bund|\hnund}\,,\\
         \label{eq.relrs2}
        &\Ber\left(\begin{array}{c:ccc}
        T^{a_1\ldots a_{i-1}b_ja_{i+1}\ldots a_r|\hmund}
        & T^{a_1\ldots a_{i-1}\hla a_{i+1}\ldots a_r|\hmund}
        & \vline
        & T^{a_1\ldots a_{i-1}\hnu_{\be}a_{i+1}\ldots a_r|\hmund} \\
                                 \hline \vphantom{\int_A^B}
         T^{*\aund|\hmu_1\ldots\hmu_{\al-1}b_j\hmu_{\al+1}\ldots \hmu_s}
         & T^{*\aund|\hmu_1\ldots\hmu_{\al-1}\hla\hmu_{\al+1}\ldots \hmu_s}
         &\vline
         & T^{*\aund|\hmu_1\ldots\hmu_{\al-1}\hnu_{\be}\hmu_{\al+1}\ldots \hmu_s} \vphantom{\int_A^B}\\
        \end{array}\right)_{{{i=1\ldots r,}\atop{j=1\ldots r-1,}}\atop{\al,\be=1\ldots s}}
        \notag\\
        &  \hspace{290pt} =(T^{\aund|\hmund})^{r+s-1}T^{b_1\ldots b_{r-1}\hla|\hnund}\,,
        \\
\label{eq.relrs3}
    &\Ber^*\left(\begin{array}{ccc:c}
        T^{a_1\ldots a_{i-1}b_ja_{i+1}\ldots a_r|\hmund} & \vline
        & T^{a_1\ldots a_{i-1} c a_{i+1}\ldots a_r|\hmund}
        & T^{a_1\ldots a_{i-1}\hnu_{\be}a_{i+1}\ldots a_r|\hmund} \\
                                 \hline \vphantom{\int_A^B}
         T^{*\aund|\hmu_1\ldots\hmu_{\al-1}b_j\hmu_{\al+1}\ldots \hmu_s} & \vline
         & T^{*\aund|\hmu_1\ldots\hmu_{\al-1} c \hmu_{\al+1}\ldots \hmu_s}
         & T^{*\aund|\hmu_1\ldots\hmu_{\al-1}\hnu_{\be}\hmu_{\al+1}\ldots \hmu_s} \vphantom{\int_A^B}
         \\
        \end{array}\right)_{{i,j=1\ldots r,}\atop{{\al=1\ldots s}\atop{\be=1\ldots s-1}}}
         \notag\\
        &  \hspace{260pt}=(T^{\aund|\hmund})^{-r-s+1}T^{*b_1\ldots b_r|c\hnu_1\ldots \hnu_{s-1}}\,.
\end{align}
Here $\aund=a_1 \dots a_r$.
\end{theorem}
(In the Berezinians in~\eqref{eq.relrs1}--\eqref{eq.relrs3},
$i$,$\al$ label   rows, while $j$,$\be$ label  columns; in the cases of~\eqref{eq.relrs2} and \eqref{eq.relrs3},
there is one column of ``wrong'' parity   in an even position in~\eqref{eq.relrs2} and in an odd position in~\eqref{eq.relrs3}.)

\subsection{Super Pl\"ucker relations  for $\Gr_{r|1}(n|1)$ in new homogeneous coordinates}
%
%
\begin{theorem} For the super Grassmannian $\Gr_{r|1}(n|1)$,
one can introduce new homogeneous 
coordinates $P^\aund, P^{*\aund}, \theta^{\aund c}$, antisymmetric in all indices, that are expressed in terms of the original super Pl\"ucker coordinates as
\begin{align}
P^\aund &:=T^{\aund|\widehat{1}} \,, \\
P^{*\aund} &:=T^{*\aund|\widehat{1}} \,, \\
\theta^{\aund c} &:=T^{*\aund|c} (P^{\aund})^2 \,.
\end{align}
In terms of these new coordinates, the super Pl\"ucker relations have the following form:
$P^{*\aund}=(P^{\aund})^{-1}$, and
\begin{align}
\label{eq:r_1_even_rel}
\sum_{i=1}^r
P^{b_i a_2 \dots a_r}
P^{\bund(b_i \leftarrow a_1)}
&= P^\aund P^\bund \,, \\
-\sum_{i=1}^r \t^{\aund b_i} P^{\bund (b_i \leftarrow c)}  + \t^{\aund c}
P^{\bund}
&=  P^\aund \, \t^{\bund c} \,,
\end{align}
where $\aund (a_i \leftarrow b_j)=a_1 \dots a_{i-1} b_j a_{i+1} \dots a_r$, and $\bund (b_i \leftarrow c)=b_1 \dots b_{i-1} c b_{i+1} \dots b_r$.
\end{theorem}
\begin{proof} In the case $\Gr_{r|1}(n|1)$, the first Pl\"ucker relation~\eqref{eq.relrs0} means that $T^{*\aund|\hat{1}}=(T^{\aund|\hat{1}})^{-1}$,
which we use below to eliminate all of the
$T^{*\aund|\hat{1}}$.  We denote $T^{\aund|\hat{1}}$   by $P^{\aund}$.

The second Pl\"ucker relation~\eqref{eq.relrs1} can be simplified as follows:
\begin{multline}
  \Ber
\begin{pmatrix}
P^{\aund(a_i \leftarrow b_j)}
 & \rvline &  \begin{matrix} 0  \\ \vdots  \\ 0 \end{matrix} \\
\hline
\begin{matrix} T^{*\aund|b_1}  & \ldots  & T^{*\aund|b_r} \end{matrix}& \rvline & T^{*\aund|\hat{1}}
\end{pmatrix} = \left(P^\aund\right)^r P^\bund \ , \\
\frac{
  \det
\begin{pmatrix}
P^{\aund(a_i \leftarrow b_j)}
  \end{pmatrix}}{T^{*\aund|\hat{1}}}
= \left(P^\aund\right)^r P^\bund \ , \\
  \det
\begin{pmatrix}
P^{\aund(a_i \leftarrow b_j)}
\end{pmatrix}
= \left(P^\aund\right)^{r-1} P^\bund \ .
\label{eq:rel2_for_r1_n1}
\end{multline}
By expanding the determinant in the left-hand side of~\eqref{eq:rel2_for_r1_n1} along the first row, we get
\begin{equation} \label{eq:plu_odd_tmp}
 \sum_{i=1}^r P^{b_i a_2 \dots a_r} C_1^i  + \t^{\aund c} \det
\begin{pmatrix}
P^{\aund(a_i \leftarrow b_j)}
\end{pmatrix} = \left(P^\aund\right)^{r-1} P^\bund \ ,
\end{equation}
where $C_1^i$ is the cofactor of the matrix 
$\begin{pmatrix}
P^{\aund(a_i \leftarrow b_j)}
\end{pmatrix}
$ 
corresponding to the determinant of the submatrix formed by deleting the first row and the $i$-th column. These cofactors can be computed by
setting $a_1=b_i$ into~\eqref{eq:rel2_for_r1_n1}: $C_1^i=(P^\aund)^{r-2}P^{\bund(b_i \leftarrow a_1)}$. Plugging this into~\eqref{eq:plu_odd_tmp},
gives~\eqref{eq:r_1_even_rel}.

The third Pl\"ucker relation~\eqref{eq.relrs2} is simplified just to $0=0$.

The fourth Pl\"ucker relation~\eqref{eq.relrs3} can be re-written as follows:
\begin{multline}
  \Ber^*
\begin{pmatrix}
P^{\aund(a_i \leftarrow b_j)}
 & \rvline &  \begin{matrix} \vdots \\ P^{\aund(a_i \leftarrow c)}   \\ \vdots \end{matrix} \\
\hline
\begin{matrix} T^{*\aund|b_1}  & \ldots  & T^{*\aund|b_r} \end{matrix}& \rvline & T^{*\aund|c}
\end{pmatrix} = \left(P^\aund\right)^r P^\bund \ , \\
\left(P^\aund \right)^{-2}
  \Ber^*
\begin{pmatrix}
P^{\aund(a_i \leftarrow b_j)}
 & \rvline &  \begin{matrix} \vdots \\ P^{\aund(a_i \leftarrow c)}   \\ \vdots \end{matrix} \\
\hline
\begin{matrix} \left(P^\aund \right)^2
T^{*\aund|b_1}  & \ldots  & \left(P^\aund \right)^2 T^{*\aund|b_r} \end{matrix}& \rvline & \left(P^\aund \right)^2 T^{*\aund|c}
\end{pmatrix} = \left(P^\aund\right)^{-r} T^{*\bund|c} \ , \\
  \Ber^*
\begin{pmatrix}
P^{\aund(a_i \leftarrow b_j)}
 & \rvline &  \begin{matrix} \vdots \\ P^{\aund(a_i \leftarrow c)}   \\ \vdots \end{matrix} \\
\hline
\begin{matrix}
\t^{\aund b_1}  & \ldots  & \t^{\aund b_r} \end{matrix}& \rvline & \t^{\aund c}
\end{pmatrix} = \left(P^\aund\right)^{-r+2} \t^{\bund c}  (P^{\bund})^{-2} \ , \\
\end{multline}
Here we used  Corollary~\ref{cor:my_trick_general_case}
to introduce new  variables $\theta^{\aund i}=T^{*\aund|i} (P^{\aund})^2$, antisymmetric  in all indices.

Using the formula for the inverse Berezinian of a wrong matrix from Theorem~\eqref{thm:my_trick}, we get:
\begin{multline}
\label{eq:rel4_r1_n1_interm}
\frac{\det \begin{pmatrix}
P^{\aund(a_i \leftarrow b_j)}
 & \rvline &  \begin{matrix} \vdots \\ P^{\aund(a_i \leftarrow c)}   \\ \vdots \end{matrix} \\
\hline
\begin{matrix}
\t^{\aund b_1}  & \ldots  & \t^{\aund b_r} \end{matrix}& \rvline & \t^{\aund c}
\end{pmatrix}}{\det^2 \begin{pmatrix}
P^{\aund(a_i \leftarrow b_j)} \end{pmatrix}}
  = \left(P^\aund\right)^{-r+2} \t^{\bund c}  (P^{\bund})^{-2} \ , \\
\det \begin{pmatrix}
P^{\aund(a_i \leftarrow b_j)}
 & \rvline &  \begin{matrix} \vdots \\ P^{\aund(a_i \leftarrow c)}   \\ \vdots \end{matrix} \\
\hline
\begin{matrix}
\t^{\aund b_1}  & \ldots  & \t^{\aund b_r} \end{matrix}& \rvline & \t^{\aund c}
\end{pmatrix}
= (P^\aund)^r \, \t^{\bund c}  \ , \\
\end{multline}
where we used second Pl\"ucker relation~\eqref{eq:rel2_for_r1_n1}. Expanding the determinant in \eqref{eq:rel4_r1_n1_interm} along the last row, we have
\begin{multline}
 \sum_{k=1}^r \t^{\aund b_k} (-1)^{r+1+k} (-1)^{k-r}
A_{00}(\aund,\bund (b_k \leftarrow c))
=  (P^\aund)^r  \, \t^{\bund c}\,, \\
-\sum_{k=1}^r \t^{\aund b_k} (P^\aund)^{r-1} P^{\bund (b_k \leftarrow c)}
=  (P^\aund)^r  \, \t^{\bund c}\,, \label{eq:rel4_r1_n1_interm2}
\end{multline}
where we introduced the notation
$
A_{00} (\aund, \bund) =
\begin{pmatrix}
P^{\aund(a_i \leftarrow b_j)} \end{pmatrix}$, and where we noticed that we can use the 
simplified form~\eqref{eq:rel2_for_r1_n1} of the second Pl\"ucker relation again.
By dividing   both sides of~\eqref{eq:rel4_r1_n1_interm2} by $(P^\aund)^{r-1}$, we arrive at the desired relation.
\end{proof}

\begin{example}
For the particular case of $\Gr_{2|1}(n|1)$, we have the even Pl\"ucker relations
\begin{equation}
   \det  \begin{pmatrix}
         P^{a_1 b_1} &  P^{a_1 b_2} \\
         P^{a_2b_1}& P^{a_2b_2}
    \end{pmatrix}
       =P^{a_1a_2} P^{b_1 b_2} \label{eq:plu2}\,,
\end{equation}
and the odd Pl\"ucker relations
\begin{equation}
  -\t^{a_1 a_2 b_1} P^{c b_2} -
    \t^{a_1 a_2 b_2}  P^{b_1 c}
  + \t^{a_1 a_2 c}  P^{b_1 b_2}
        = \t^{b_1 b_2 c} P^{a_1 a_2}\,. \label{eq:plu4}
\end{equation}

\end{example}

\section{Super cluster structure on the super Grassmannian of $2|0$-planes in $n|1$-space}

\label{sec:supercluster}

Here is   further investigation of the proposed super cluster structure for the Super Grassmannian $\Gr_{2|0}(n|1)$, which was given in broad brush at the end of~\cite{shemya:voronov2019:superplucker}.

\subsection{Super clusters and  mutations}

\subsubsection{\textbf{Decorated triangulations and super clusters.}}


The cluster structure that is proposed in~\cite{shemya:voronov2019:superplucker} for  Super Grassmannians $\Gr_{2|0}(n|1)$ is a ``decoration'' of the standard construction for the Grassmannians $\Gr_{2}(n)$. For $\Gr_{2|0}(n|1)$, the clusters are parameterized by  triangulations of a regular $n$-gon that are in addition decorated in the following way: one of the proper diagonals has its end-points marked, see e.g. Fig.~\ref{fig:pent_marked} where the red dots show the marked endpoints.

\begin{figure}[ht]
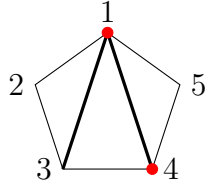

 \begin{center}
\pent
\end{center}
\caption{Decorated triangulation corresponding to cluster $\{ T^{13} , T^{14} | \theta^1 , \theta^4 \}$ for $\Gr_2(5)$.
}
\label{fig:pent_marked}
\end{figure}
Every cluster has the form
$$
 \{T^{ab}, \underbrace{T^{cd}, \dots , T^{kl}}_{n-4} | \theta^{a},
 \theta^{b} \} \, .
$$
As in the classical case, here $T^{ab}, T^{cd}, \dots , T^{kl}$ are proper non-intersecting diagonals of an $n$-gon. They are even cluster variables. The indices  are in strictly increasing order, $a<b$, $c<d, \dots, k<l$. Every cluster contains two odd variables $\theta^a, \theta^b$ such that $T^{ab}$ is in the cluster.


One can go through all such decorated triangulations/clusters starting with any of them (seed) by means
of the following two types of mutations. \emph{Even mutation} mutates at a decorated diagonal, in which case one even and two odd variables that are represented by this diagonal would mutate in an ensemble; \emph{odd mutation} mutates at a decorated vertex, in which case the odd variable corresponding to it may be mutated into another odd variable such that the marked vertices are still the end-points of a proper diagonal.

\subsubsection{\textbf{Odd mutation.}}

Given a cluster, let its odd cluster variables
be denoted $\theta^{a}, \theta^{b}$. Then each of them, say $\theta^a$, can mutate into $\theta^{c}$ provided $T^{bc}$ or $T^{cb}$ is in the cluster (so that the new marked points will be the end-points of some proper diagonal).
See, e.g. Fig.~\ref{fig:pent_odd}.
\begin{figure}[ht]
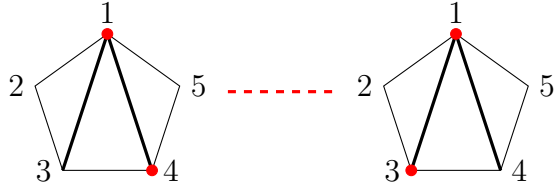

 \begin{center}
\pentodd
\end{center}
\caption{Example of an odd mutation: $\theta^4$ mutates into $\theta^3$.}
\label{fig:pent_odd}
\end{figure}
Such mutation is governed by the corresponding super Pl\"ucker relation. The signs will be defined uniquely once the vertices of the $n$-gon are ordered. Indeed, we start with the relation in its general form~\eqref{eq:second_mutation_rule},
\begin{equation} \label{eq:main_odd_mut_rel}
  T^{ab} \theta^{c} + \theta^a T^{bc}  = T^{ac}\theta^{b}
\end{equation}
where we add condition $a<b<c$ to not contradict our new requirement that the indices are strictly increasing. Now, if $\theta^a$ mutates into $\theta^{c}$ and $a<b<c$, super Pl\"ucker relation
~\eqref{eq:main_odd_mut_rel} implies
\begin{equation} \label{eq:}
   \theta^{c}  = \frac{T^{ac}\theta^{b} - \theta^a T^{bc}}{T^{ab}} \ ,
\end{equation}
and so the sign in the mutation relation is negative. On the other hand, if again $\theta^a$ mutates into $\theta^{c}$ but $a<c<b$,
super Pl\"ucker relation
~\eqref{eq:main_odd_mut_rel} implies
\begin{equation} \label{eq:2}
   \theta^{c}  = \frac{T^{ac}\theta^{b} + \theta^a T^{cb}}{T^{ab}} \ ,
\end{equation}
and so the sign in the mutation relation is positive.

One can show that by applying successive odd mutations, it is possible to make any diagonal of a given triangulation marked. 

\subsubsection{\textbf{Even mutation.}}

Given a cluster, let its odd cluster variables
be denoted $\theta^{a}, \theta^{b}$, $a<b$. Then it must contain a matching even cluster variable $T^{ab}$ represented by a proper diagonal of a triangulation. This marked diagonal can be mutated into another diagonal with marked
endpoints via a flip as in the classical case, where in addition the marked points are also moved together with the diagonal forming another marked diagonal. See, e.g. Fig.~\eqref{fig:pent_even}.
\begin{figure}[ht]
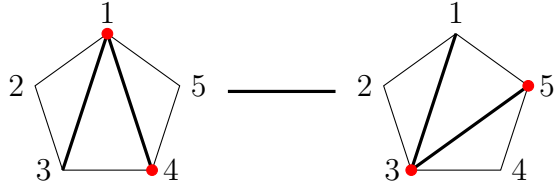

 \begin{center}
\penteven
\end{center}
\caption{Example of an even mutation: variables $\{T^{14}, \t^1, \t^4 \}$ mutate into $\{T^{35}, \t^3, \t^5 \}$.}
\label{fig:pent_even}
\end{figure}
%
%
%
Even mutation is governed by relations~\eqref{eq:first_mutation_rule} and~\eqref{eq:second_mutation_rule}.
When one expresses new cluster variables in terms of old cluster, the signs are uniquely defined by the order of the vertices' labels.

To illustrate this, let us suppose the marked diagonal corresponding to
cluster variables $\theta^{a}, \theta^{b}$, $T^{ab}$ mutates into the marked diagonal corresponding to $\theta^{c}, \theta^{d}$, $T^{cd}$, and let $a<b<c<d$.
In this case, ~\eqref{eq:first_mutation_rule} implies $T^{ac} T^{bd} = T^{ab} T^{cd} + T^{ad} T^{bc}$ and
\begin{equation}
 T^{cd} = \frac{T^{ac} T^{bd} -T^{ad} T^{bc} }{T^{ab}} \ .
\end{equation}
For the odd variables,~\eqref{eq:second_mutation_rule} implies $T^{ab} \theta^{c} + \theta^a T^{bc}  = T^{ac}\theta^{b}$ and
$T^{ab} \theta^{d} + \theta^a T^{bd}  = T^{ad}\theta^{d}$, and so
\begin{align}
  \theta^{c}  &= \frac{T^{ac}\theta^{b} - \theta^a T^{bc}}{T^{ab}} \ , \\
 \theta^{d}  &= \frac{T^{ad}\theta^{d} + \theta^a T^{bd}}{T^{ab}} \ .
\end{align}
%
%
%
%

\begin{remark}
  Since any diagonal in a triangulation can be made marked by applying odd mutations, there is no need to introduce yet another type of   mutations for which there is no change of odd variables and only an even variable  is  mutated. Such a ``purely even'' mutation would be the composition of a certain number of odd mutations and an even mutation associated with a marked diagonal.
\end{remark}

\subsection{The exchange graph}

The first non-trivial example is the exchange graph for $\Gr_{2|0}(5|1)$, see Fig.~\ref{fig:exch_Gr2_51}.
\begin{figure}[ht]
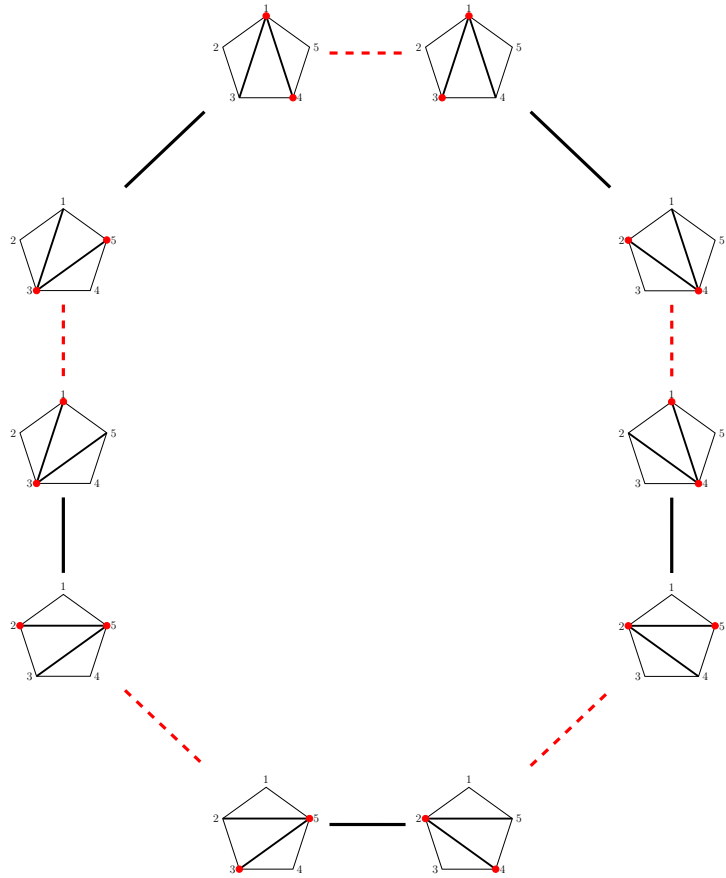

 \begin{center}
  \pentaexch
 \end{center}
\caption{Exchange graph for cluster structure for $\Gr_{2|0}(5|1)$.}
\label{fig:exch_Gr2_51}
\end{figure}

In general, the exchange graphs for the cluster structures corresponding to $\Gr_{2|0}(n|1)$ can be obtained from the exchange graphs for $\Gr_{2}(n)$, by replacing every vertex with a set of $n-3$ vertices that correspond to the same triangulation but different marking. For example, see exchange graphs for $\Gr_{2}(5)$ and $\Gr_{2|0}(5|1)$ on Fig.~\ref{fig:exch_Gr2_51}.

\begin{figure}[ht]
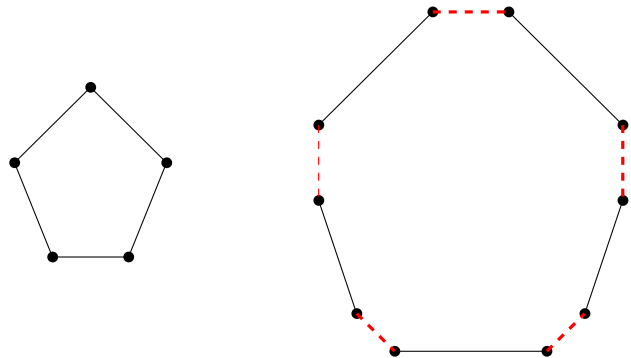

 \begin{center}
  \pentaexchfiveboth
 \end{center}
\caption{Exchange graphs for cluster structures corresponding to  $\Gr_{2}(5)$ and  $\Gr_{2}(5|1)$. The latter is obtained from the former by blowing up every vertex, and replacing this vertex by two vertices (which corresponds to two different markings for every triangulation).}
\label{fig:exch_Gr2_51_both}
\end{figure}

Also see exchange graphs for $\Gr_{2}(6)$ and $\Gr_{2|0}(6|1)$ on Fig.~\ref{fig:exch_Gr2_6}.
\begin{figure}[ht]
 \begin{center}
\pentaexchsix
\pentaexchsupersix
 \end{center}
\caption{Exchange graphs for cluster structures corresponding to  $\Gr_{2}(6)$ and  $\Gr_{2}(6|1)$. The latter is obtained from the former by blowing up every vertex, and replacing this vertex by three vertices (which corresponds to three different markings for every triangulation).}
\label{fig:exch_Gr2_6}
\end{figure}

Fig.~\ref{fig:exch_Gr2_6_top} gives details on how an insert looks like for the case $\Gr_{2}(6|1)$.
\begin{figure}[ht]
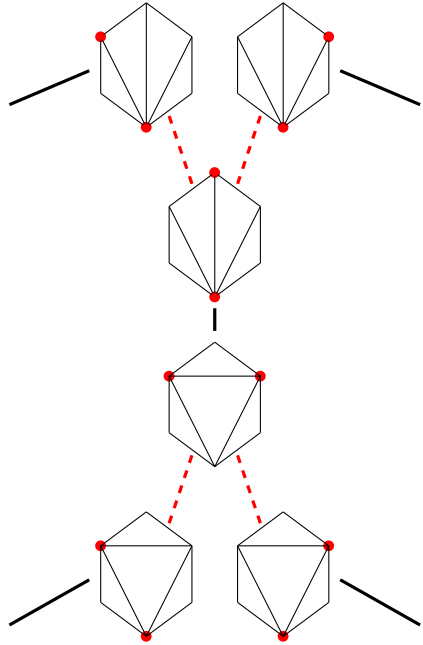

 \begin{center}
\pentaexchsixtop
\end{center}
\caption{Example of mutations in the cluster structure corresponding to $\Gr_{2}(6|1)$. This picture corresponds to the very top part of the left exchange graph on Fig.~\ref{fig:exch_Gr2_6}.}
\label{fig:exch_Gr2_6_top}
\end{figure}

\newpage
\section{Modification of Penner-Zeitlin's super Ptolemy equations}

\label{sec:superptol}

An alternative approach for construction of super cluster algebras is based on Penner-Zeitlin's   super Ptolemy equation~\cite{pennerzeitlin2019decorated} for the decorated super Teichm\"{u}ller space $\widetilde{\mathcal{ST}}$.
This super Ptolemy equation is used as an exchange relation in the  Musiker-Ovenhouse-Zhang~\cite{musiker2021expansion} version of super cluster structure for $\widetilde{\mathcal{ST}}$.

Here we propose a re-writing of Penner-Zeitlin's super Ptolemy equation in a new form. It may be useful for  the study of the super cluster structure of the decorated super Teichm\"{u}ller space $\widetilde{\mathcal{ST}}$ and for comparison of that type of super cluster algebras with those arising from super Grassmannians.

Recall that  the classical \textit{Teichm\"uller space} $\mathcal{T}$ is a moduli space for hyperbolic metrics on a (topological) surface  $F$.
Penner's \textit{decorated  Teichm\"uller space}~\cite{Penner_main_1987} is a trivial bundle $\widetilde{\mathcal{T}} \rightarrow \mathcal{T}$  with fiber $\RR^n_{>0}$, where fiber coordinates are the numbers assigned to each puncture that fix the heights of the horocycles. These additional data
allow  to consider the ``regularized lengths''  $\ell$ for the geodesics with the endpoints at the punctures (as the lengths of the truncated parts) and to introduce the corresponding \emph{$\lambda$-lengths} as $\lambda=exp(\ell/2)$.
Penner's construction equips the space $\widetilde{\mathcal{T}}$ with global coordinate systems  that are enumerated by ``ideal'' triangulations of $F$, the vertices of which are the punctures. For a given triangulation, the $\lambda$-lengths of the edges are the coordinates. Transforming one triangulation into another by a flip  leads to a  changes of $\lambda$-lengths   governed by the Ptolemy equation
$ef=ac+bd$. Note that the Ptolemy relation and the (classical) Pl\"{u}cker relation for $\Gr_2(4)$ have exactly the same form.

In 2019, Penner-Zeitlin~\cite{pennerzeitlin2019decorated} introduced global coordinates for the connected components of the \textit{decorated super Teichm\"uller space} $\widetilde{\mathcal{ST}}$, where in addition to even coordinates as above, there are odd coordinates  known as \emph{$\mu$-invariants},  and which correspond  to the faces of a triangulation. Evolution from one triangulation to another is governed by flips and the following \emph{super Ptolemy equations}:
\begin{align}
ef &= (ac+bd)\left(1 + \frac{\sigma \theta \sqrt{Z}}{1+ Z} \right)\,, \label{rel:1} \\
\theta' &=  \frac{\theta + \sigma \sqrt{Z}}{\sqrt{1+ Z}}\,, \\
\sigma' &=  \frac{ \sigma - \theta \sqrt{Z}}{\sqrt{1+ Z}}\,,
\end{align}
where $Z=ac/bd$. Here  the even variables $a,b,c,d$  correspond to the sides of a quadrilateral, while  $e$ and $f$ correspond to the diagonals before and after the flip, respectively. This is in the same way as above, but the exchange relation~\eqref{rel:1} now involves a nilpotent extra  term depending on odd variables.  The odd variables $\sigma$ and $\theta$ correspond to the faces before the flip, and   $\sigma'$ and $\theta'$ are the corresponding mutated ones (after the flip). It is easy to see that $\sigma \theta = \sigma' \theta'$.

What new that we propose here for $\widetilde{\mathcal{ST}}$, is as follows.
It turns out that it is possible to introduce new even variables corresponding to the diagonals and show that the non-standard form of the even part of the  Penner-Zeitlin super Ptolemy relation~\eqref{rel:1},   in the new variables becomes    the usual Ptolemy (or the usual Pl\"ucker) relation.
\begin{lemma}
Define new even variables:
 \begin{align*}
\bar{e} &:=e \left( 1 + \frac{\sigma \theta \sqrt{Z}}{1 + Z} \right)^{1/2}\,, \\
\bar{f} &:=f \left( 1 + \frac{\sigma' \theta' \sqrt{Z}}{1 + Z} \right)^{1/2}\,
\end{align*}
(note $\sigma \theta = \sigma' \theta'$). Then the even super Ptolemy relation~\eqref{rel:1} 
becomes
\begin{equation}
  \bar{e} \bar{f} = ac + bd\,.
 \end{equation}
\end{lemma}
\begin{proof}
 Compute $\bar{e} \bar{f} =ef\left( 1 + \frac{\sigma \theta \sqrt{Z}}{1 + Z} \right)$.
\end{proof}

\section*{Acknowledgments}
This work has come into being thanks to many interesting conversations at the Newton Institute special semester in Cluster algebras and representation theory
with Michael Shapiro, Michael Gekhtman, Gregg Musiker, and last but not the least, Alastair King. 

This work was partially supported by Simons Foundation collaboration grant for Mathematicians number 846970 and a London Mathematical Society Scheme 2 grant.


\end{document}